\documentclass[a4paper,UKenglish,cleveref, autoref,thm-restate]{lipics-v2021}

\usepackage[utf8]{inputenc}
\usepackage{hyperref}
\usepackage{graphicx}
\usepackage{enumerate}
\usepackage[linesnumbered,lined,boxed,commentsnumbered]{algorithm2e}
\usepackage{todonotes}
\usepackage{stmaryrd}
\usepackage{bbm}
\usepackage{enumitem}
\usepackage{multidef,xspace}
\usepackage{thmtools}
\usepackage{xspace}
\usepackage{amsmath,amssymb,amsfonts}
\usepackage[normalem]{ulem}

\usepackage{xcolor,pgf,tikz}
\usetikzlibrary{arrows,automata,shapes,decorations,calc,positioning}
\usepackage{circuitikz}

\font\dsrom=dsrom10 scaled 1200
\def \indic{\textrm{\dsrom{1}}}

\usepackage{macros}

\newcommand{\nat}{\mathbb{N}}

\newcommand{\real}{\mathbb{R}}

\newcommand{\EvalRP}{{\upshape \textsf{EvalRP}}\xspace}
\newcommand{\EstimRP}{{\upshape \textsf{EstimRP}}\xspace}

\newcommand{\EvalEV}{{\upshape \textsf{EvalER}}\xspace}
\newcommand{\EstimEV}{{\upshape \textsf{EstimER}}\xspace}

\newif\ifARXIV
\ARXIVtrue

\pdfoutput=1 

\ifARXIV
\nolinenumbers
\hideLIPIcs  
\fi

\title{Beyond Decisiveness of Infinite Markov Chains}

\author{Benoît Barbot}{Univ Paris Est Creteil, LACL, F-94010 Creteil, France}{benoit.barbot@lacl.com}{https://orcid.org/0000-0003-2417-3064}{}
\author{Patricia Bouyer}{Université Paris-Saclay, CNRS, ENS
  Paris-Saclay, Laboratoire Méthodes Formelles, 91190 Gif-sur-Yvette, France}{patricia.bouyer@lmf.cnrs.fr}{0000-0002-2823-0911}{}
\author{Serge Haddad}{Université Paris-Saclay, CNRS, ENS Paris-Saclay,
  Laboratoire Méthodes Formelles, 91190 Gif-sur-Yvette, France}{serge.haddad@lmf.cnrs.fr}{https://orcid.org/0000-0002-1759-1201}{}
\authorrunning{B. Barbot, P. Bouyer and S. Haddad}

\Copyright{Benoît Barbot, Patricia Bouyer and Serge Haddad} 



\ifARXIV
\else
\relatedversion{\pat{The full version of this paper is available as~\cite{ARXIV}.}}
\fi

\funding{This work has been partly supported by ANR projects
    MAVeriQ (ANR-20-CE25-0012) and BisoUS (ANR-22-CE48-0012).}

\EventEditors{Siddharth Barman and S{\l}awomir Lasota}
\EventNoEds{2}
\EventLongTitle{44th IARCS Annual Conference on Foundations of Software Technology and Theoretical Computer Science (FSTTCS 2024)}
\EventShortTitle{FSTTCS 2024}
\EventAcronym{FSTTCS}
\EventYear{2024}
\EventDate{December 16--18, 2024}
\EventLocation{Gandhinagar, Gujarat, India}
\EventLogo{}
\SeriesVolume{323}
\ArticleNo{9}

\begin{document}
\maketitle
\begin{abstract}
  Verification of infinite-state Markov chains is still a challenge
  despite several fruitful numerical or statistical approaches. For
  \emph{decisive} Markov chains, there is a simple numerical algorithm
  that frames the reachability probability as accurately as required
  (however with an unknown complexity). On the other hand when
  applicable, statistical model checking is in most of the cases very
  efficient. Here we study the relation between these two approaches
  showing first that decisiveness is a necessary and sufficient
  condition for almost sure termination of statistical model
  checking. Afterwards we develop an approach with application to both
  methods that substitutes to a non decisive Markov chain a decisive
  Markov chain with the same reachability probability. This approach
  combines two key ingredients: abstraction and importance sampling (a
  technique that was formerly used for efficiency). We develop this
  approach on a generic formalism called layered Markov chain (LMC).
  Afterwards we perform an empirical study on probabilistic pushdown
  automata (an instance of LMC) to understand the complexity factors
  of the statistical and numerical algorithms. To the best of our
  knowledge, this prototype is the first implementation of the
  deterministic algorithm for decisive Markov chains and required us
  to solve several qualitative and numerical issues.
\end{abstract}

\section{Introduction}

\subparagraph{Infinite-state discrete time Markov chains.}

In finite Markov chains, computing reachability probabilities can be
performed in polynomial time using linear algebra
techniques~\cite{Rutten04}.  The case of infinite Markov chains is
much more difficult, and has initiated several complementary
proposals:
\begin{itemize}[nosep]
  \item A first approach consists in analyzing the high-level
        probabilistic model that generates the infinite Markov chains.  For
        instance in~\cite{Esparza06}, the authors study probabilistic
        pushdown automata and show that the reachability probability can be
        expressed in the first-order theory of the reals.  Thus (by a
        dichotomous algorithm) this probability can be approximated within
        an arbitrary precision.
  \item A second approach consists in designing algorithms, whose
        correctness relies on a semantic property of the Markov chains, and
        which outputs an interval of the required precision that contains
        the reachability probability. \emph{Decisiveness}~\cite{ABM07} is
        such a property: given a target set $T$, it requires that almost
        surely, a random path either visits $T$ or some state from which $T$
        is unreachable. In order to be effective, this algorithm needs the
        decidability of the (qualitative) reachability problem. For instance
        finite Markov chains are decisive w.r.t. any set of states.  Several
        other classes of denumerable Markov chains are decisive by
        construction: Petri nets (or equivalently VASS) with constant
        weights\footnote{That is, each transition is assigned a weight, and
          the probability for a transition to be fired is its relative
          weight w.r.t. all enabled transitions.} on transitions w.r.t. any
        upward-closed target set~\cite{ABM07}, lossy channel systems with
        constant weights and constant message loss probability~\cite{ABM07}
        w.r.t. any finite target set, regular Petri nets with arbitrary
        weights w.r.t. any finite target set~\cite{FHY23}.  Also a critical
        associated decision problem is the decidability of decisiveness in
        the high-level model that generates the Markov chains. Decisiveness
        is decidable for several classes of systems: probabilistic pushdown
        automata with constant weights~\cite{Esparza06}, random walks with
        polynomial weights~\cite{FHY23} which can be generalized to
        probabilistic homogeneous one-counter machines with polynomial
        weights~\cite{FHY23}. This class is particularly interesting since
        it extends the well-known model of quasi-birth death processes
        (QBDs).
  \item The statistical model checking (SMC)~\cite{YS02,YS06} approach consists in
        generating numerous random paths and computing an interval of the required
        precision that contains the reachability probability with an arbitrary high
        probability (the confidence level).  As we will show later on, the
        effectiveness of SMC also requires some semantic property,
        which will happen to be  decisiveness.
\end{itemize}

\subparagraph{SMC and importance sampling.} When the reachability
probability is very small, the SMC approach requires a huge number of
random paths, which prohibits its use. In order to circumvent this
problem (called the \emph{rare event problem}), several approaches
have been proposed (see for instance~\cite{RT2009}), among which the
\emph{importance sampling} method. This seems to be in practice, one
of the most efficient approaches to tackle this problem. Importance
sampling consists in sampling the paths in a biased\footnote{In the
  sense that probability values in the biased chain differ from the
  original chain.} Markov chain (w.r.t. the original one) that
increases the reachability probability. In order to take into account
the bias, a likelihood for any path is computed (on-the-fly) and the
importance sampling algorithm returns the empirical average value of
the likelihood. While the expected returned value is equal to the
reachability probability under evaluation, the confidence interval
returned by the algorithm is, without further assumption, only
``indicative'' (i.e., it does not necessary fulfill the features of a
confidence interval) -- boundedness of the likelihood is indeed
required (but hard to ensure). In~\cite{BHP12}, a simple relation
between the biased and the original finite Markov chain is stated that
(1) ensures that the confidence interval returned by the algorithm is
a ``true'' interval and that (2) the variance of the estimator (here
the likelihood) is reduced w.r.t. the original estimator, entailing an
increased efficiency of the SMC.

\subparagraph{Related work when the Markov chain is not decisive.}
Very few works have addressed the effectiveness of SMC for infinite
non decisive Markov chains. The main proposal~\cite{YounesCZ10}
consists in stopping the computation at each step with some fixed
(small) probability.  The successful paths are equipped with a
numerical value, whose average over the paths is returned by the
algorithm. It turns out that it is an importance sampling method,
which has surprisingly not been pointed out by the authors. However
here again, since the likelihood is (in general) not bounded, the
interval returned by the algorithm is not a confidence interval.  An
alternative notion called divergence has been proposed
in~\cite{FHY23b} to partly cover the case of non-decisive Markov
chains.

\subparagraph{Our contributions.} We introduce the \emph{reward
  reachability problem} (a slight generalization of the reachability
problem) by associating a reward with every successful path and
looking for the expected reward. We first establish that decisiveness
is a necessary and sufficient condition for the almost-sure
termination of SMC for bounded rewards.
\begin{itemize}[nosep]
  \item Our major contribution consists in establishing a relation
        between a non-decisive Markov chain and an auxiliary Markov chain,
        called an \emph{abstraction}, with the following property: they
        can be combined into a biased Markov chain, which happens to be
        \textbf{decisive};  the SMC with importance sampling on this chain
        provides a confidence interval for the reachability probability of
        the original Markov chain.
  \item We furthermore show that importance sampling can be applied to
        adapt (based on the abstraction) the deterministic algorithm
        of~\cite{ABM07}.
  \item Afterwards we illustrate the interest of
        this approach, by exhibiting a generic model called \emph{layered Markov
          chains} (LMC), which can be instantiated for instance by
        probabilistic pushdown automata with polynomial weights. These
        automata cannot be handled with the technique of~\cite{Esparza06}.
  \item Finally we present several experiments, based on the tool
        Cosmos~\cite{refcosmos}, which compare the SMC and the deterministic
        approaches. It allows to identify how various factors impact the
        efficiency of the  algorithms. We provide within the tool
        Cosmos the first implementation of the deterministic approach for
        decisive Markov chains, which required us to solve several
        numerical issues. As a rough summary, at the price of a confidence
        level against certainty, the computing time of SMC is generally
        several magnitude orders smaller than the one of the deterministic
        algorithm.
\end{itemize}

\subparagraph{Organization.} In section~\ref{sec:preliminaries}, we
introduce the numerical and statistical specification of the reward
reachability problem, and we recall the notion of decisiveness. In
section~\ref{sec:decisive}, we focus on the decisiveness property
establishing that decisiveness is a necessary and sufficient condition
for almost sure termination of statistical model checking.
Section~\ref{sec:beyond} contains our main contribution: the
specification of an abstraction of a Markov chain, its use for solving
the reward reachability problems for non decisive Markov chains via
importance sampling and the development of this method for LMCs.
Afterwards in Section~\ref{sec:experiments}, we present some
implementation details and experimentally compare the deterministic
and the statistical approaches. We conclude and give some perspectives
to this work in Section~\ref{sec:conclusion}.

\ifARXIV Missing proofs and more details on the implementation can
  found in the appendix.
\fi

\section{Preliminaries}
\label{sec:preliminaries} \label{sec:prelim}


In this preliminary section we define Markov chains, and the
decisiveness property.

\begin{definition}
  \label{def:dtmc}
  A \emph{discrete-time Markov chain} (or simply \emph{Markov chain})
  $\calC = (S,\Pt)$ is defined by a countable set of states
  $S$ and a transition probability matrix $\Pt$ of size $S \times S$.
  Given an initial state $s_0 \in S$, the state of the chain at time $n$ is a
  random variable (r.v. in short) $X^{\calC,s_0}_n$ defined by:
  $\Pr\big(X^{\calC,s_0}_0=s_0\big)=1$ and
  $\Pr\big(X^{\calC,s_0}_{n+1}=s' \mid \bigwedge_{i\leq
    n}X^{\calC,s_0}_i=s_i\big)=\Pt(s_n,s')$.
\end{definition}

If $\calC = (S,\Pt)$ is a Markov chain, we write
$E_\calC = \{(s,s') \in S \times S \mid \Pt(s,s')>0\}$ for the set of
edges of $\calC$, and $\to_\calC$ for the corresponding edge
relation. A state $s$ is \emph{absorbing} if $\Pt(s,s)=1$.
A Markov chain $\calC$ is said \emph{effective} whenever for every
$s\in S$, the support of $\Pt(s,\cdot)$ is finite and computable, and
for every $s,s' \in S$, $\Pt(s,s')$ is computable. A target set
$T \subseteq S$ is said \emph{effective} whenever its membership
problem is decidable.  In the following, we will always consider
effective Markov chains and effective target sets when speaking of
algorithms, without always specifying it.

A finite (resp. infinite) path is a finite (resp. infinite) sequence
of states $\rho = s_0 s_1 s_2 \ldots \in S^+$ (resp. $S^\omega$) such
that for every $0 \le i$, $(s_i,s_{i+1}) \in E_\calC$. We write
$\first(\rho)$ for $s_0$, and whenever $\rho \in S^+$, we write
$\last(\rho)$ for the last state of $\rho$. For every $n \in \bbN$, we
write $\rho[n] \eqdef s_n$ and
$\rho_{\le n} \eqdef s_0 s_1 s_2 \ldots s_n$.  If
$\rho=s_0\ldots s_n$, $\Pr(\rho)$ is equal to
$\prod_{i<n}\Pt(s_i,s_{i+1})$ and corresponds to the probability that
this path is followed when starting from its initial state $s_0$.

The random infinite path generated by process
$(X^{\calC,s_0}_n)_{n \in \bbN}$ will be denoted
$\varrho^{\calC,s_0}$.  Note that $s \to^*_{\calC} s'$ if and only if
$\Pr \big(\varrho^{\calC,s} \models \Diamond \{s'\} \big)>0$ (we use
the $\Diamond$ modality of temporal logics, which expresses
\textit{Eventually}, and later, we will also write $\Diamond_{>0}$ for
the \textit{strict Eventually} modality --eventually but not now--, as
well as $\Diamond_{\le n}$ for \textit{$n$-steps Eventually}).
Finally, for every $s,s' \in S$, we define the time from $s$ to $s'$
as the random variable
$\tau^{\calC,s,s'} = \min\{i \in \bbN \mid i>0\ \text{and}\
X_i^{\calC,s} = s'\}$, with values in $\bbN_{>0} \cup \{+\infty\}$.
To ease the reading, we will omit subscripts $_{\calC}$ and $_T$, or
superscripts $^{\calC}$ in the various notations, whenever it is
obvious in the context.

In this paper we are interested in evaluating the probability to reach
a designed target set $T$ from an initial state $s_0$ in a Markov
chain $\calC$, that is,
$\mu_{\calC,T}(s_0) \eqdef \Pr \big(\varrho^{\calC,s_0} \models
\Diamond T\big)$. \ifARXIV
In general, it might be difficult to compute such a
value, which will often not even be a rational number (see
Appendix~\ref{app:termination}). \fi
That is why like many other research works
we will show how to compute accurate approximations (surely or with
a high level of confidence).  We present our solutions in a more general
setting which would anyway be necessary in the following developments.

\begin{definition}
  Let $T\subseteq S$ and $\rho \in S^\omega$. We let
  $\first_T(\rho) := \min \{i \in \bbN \mid \rho[i] \in T\} \in \bbN
  \cup \{\infty\}$.  Let $L : S^+ \to \bbR$ be a function.  The
  function $f_{L,T} : S^\omega \to \bbR$ is then defined
  by:\footnote{This function is measurable as a pointwise limit of
    measurable functions.}
  \[
    f_{L,T}(\rho) := \left\{\begin{array}{ll} L\big(\rho_{\le
                             \first_T(\rho)}\big) & \text{if}\ \first_T(\rho) \in \bbN \\
                           0 & \text{otherwise} \end{array} \right.
  \]
\end{definition} 

We say that $f_{L,T}(\rho)$ is the \emph{reward} of $\rho$. The
function $f_{L,T}$ is called the \emph{$T$-function for $L$}; let
$B \in \bbR_{>0}$, $f_{L,T}$ is said \emph{$B$-bounded} whenever
$\max(|f_{L,T}(\rho)|\mid \rho \in S^\omega) \leq B$. Observe that
$f_{L,T}$ could be $B$-bounded for some $B$ even if $L$ is unbounded.

We will be interested in evaluating the expected reward
$\nu_{\calC,L,T}(s_0) \eqdef
\Es\big(f_{L,T}(\varrho^{\calC,s_0})\big)$.\footnote{$\Es$ denotes the
  expectation.} Note that if $L$ is constant equal to $1$, then
$f_{L,T} = \indic_{\Diamond T}$ is the indicator function for paths
that visit $T$, in which case
$\nu_{\calC,L,T}(s_0) = \mu_{\calC,T}(s_0)$.

We define two problems related to the accurate
estimation of these values:
\begin{itemize}[nosep]
\item The \EvalEV problem (\EvalEV stands for ``Evaluation of the
  Expected Reward'') asks for \emph{a deterministic algorithm}, which:
  \begin{enumerate}[nosep]
  \item takes as input a Markov chain $\calC$, an initial state $s_0$,
    a computable function $L : S^+ \to \bbR_{\ge 0}$, a target set
    $T$, a precision $\varepsilon>0$, and
  \item outputs an interval $I \subseteq \bbR$ of length bounded by
    $\varepsilon$ such that $\nu_{\calC,L,T}(s_0) \in I$.
  \end{enumerate} 
  The particular case of the
  reachability probability (when $L$ is constant equal to $1$) is
  denoted \EvalRP.
\item The \EstimEV problem (\EstimEV stands for ``Estimation of the
  Expected Reward'') asks for a \emph{probabilistic Las Vegas algorithm},
  which:
  \begin{enumerate}[nosep] 
  \item takes as input a Markov chain $\calC$, an initial state $s_0$,
    a computable function $L : S^+ \to \bbR_{\ge 0}$, a target set
    $T$, a precision $\varepsilon>0$, a confidence value $\delta>0$, and
  \item outputs a random interval $I \subseteq \bbR$ of length bounded
    by $\varepsilon$ such that
    $\Pr\big(\nu_{\calC,L,T}(s_0) \notin I\big) \leq \delta$, and
    $\Es(\textsf{mid}(I)) = \nu_{\calC,L,T}(s_0)$, where
    $\textsf{mid}(I)$ is the middle of interval $I$.\footnote{The last
      condition on the middle of $I$ means that the estimator is
      unbiased.}
  \end{enumerate}  
  The particular case of
  the reachability probability (when $L$ is constant equal to $1$) is
  denoted \EstimRP.
\end{itemize}

\medskip In~\cite{ABM07}, the concept of decisiveness for Markov
chains was introduced. Roughly, decisiveness allows to lift some
``good'' properties of finite Markov chains to countable Markov
chains. We recall this concept here.  Let $T\subseteq S$ and denote
the ``avoid set'' of $T$ by
$\Av_{\calC}(T) \eqdef \{s \in S \mid \Pr(\varrho^{\calC,s} \models
\Diamond T) = 0\}$.
\begin{definition}
  The Markov chain $\calC$ is \emph{decisive w.r.t. $T$ from $s_0$} if
  $\Pr\big(\varrho^{\calC,s_0} \models \Diamond T \vee \Diamond
  \Av_{\calC}(T)\big) =1$.
\end{definition}


\section{Analysis of decisive Markov chains}
\label{sec:decisive}


We fix for this section a Markov chain $\calC = (S,\Pt)$, an initial
state $s_0$, a computable function $L : S^+ \to \bbR_{\ge 0}$ and a
target set $T$, and we assume w.l.o.g. that $T$ is a set of absorbing
states.  We present two approaches (extended from the original ones)
to compute the expected value of the function $f_{L,T}$ that require
$\calC$ to be decisive w.r.t. $T$ from $s_0$.

\subsection{Decisiveness and approximation algorithm}

In the original paper proposing the concept of
decisiveness~\cite{ABM07}, ``theoretical'' approximation schemes were
designed. We slightly extend the one designed for reachability
objectives in our more general setting, see
Algorithm~\ref{algo:approx}.

\begin{algorithm}
  \SetKwInOut{Input}{input}\SetKwInOut{Output}{output}
  \Input{$\calC=(S,\Pt)$ a countable Markov chain, $s_0 \in S$ an
    initial state, $L : S^+ \to \bbR_{\ge 0}$ a computable function,
    $T \subseteq S$ a target set s.t. $\Av_\calC(T)$ is effective and
    $f_{L,T}$ is $B$-bounded, $\varepsilon>0$ a precision.}  $e := 0$,
  $p_{\text{fail}}:=0$, $p_{\text{succ}}:=0$;
  $\textit{set} := \{(1,s_0)\}$\;
  \While{$1-(p_{\text{succ}}+ p_{\text{fail}})>\varepsilon/2B$} {
    $(p,\rho) := {\sf fair\_extract}(\textit{set})$;
    $s := \last(\rho)$\;
    \lIf {$s \in T$} {$e := e+p \cdot L(\rho)$;
      $p_{\text{succ}} := p_{\text{succ}} + p$}
    \lElseIf {$s \in \Av(T)$} {$p_{\text{fail}} :=
      p_{\text{fail}} + p$}
    \uElse {
      \lFor {$s \to_\calC s'$}
      {${\sf insert}(set,(p \cdot \Pt(s,s'),\rho s'))$}
    }
  }
  \Return {$[e-\varepsilon/2,e+ \varepsilon/2]$}
  \caption{Approximation scheme for the \EvalEV problem; the
    $ {\sf fair\_extract}$ operation ensures that any element put in
    the set cannot stay forever in an execution including an infinite
    number of extractions; a simple implementation can be done with a
    queue.\label{algo:approx}}
\end{algorithm}

The termination and correctness of this algorithm is established by
the following proposition, whose proof \ifARXIV is postponed to
Appendix~\ref{app:termination} \fi and a special case of which is given in~\cite{ABM07}.

\begin{restatable}[Termination and correctness of
  Algorithm~\ref{algo:approx}]{proposition}{algodec}
  \label{coro:algo:approx}
  Algorithm~\ref{algo:approx} solves the \EvalEV problem if and only
  if $\calC$ is decisive w.r.t. $T$ from $s_0$.
\end{restatable}

Up to our knowledge, the version of this algorithm for computing
reachability probabilities has not been implemented, hence the
terminology ``theoretical'' scheme above. Also, there is no known
convergence speed. Later in section~\ref{sec:experiments}, we briefly
describe an efficient implementation of this scheme
by designing some tricks.

\subsection{Decisiveness and (standard) statistical model-checking}
\label{subsec:decisive_smc}

The standard \emph{statistical model-checking} (SMC in short) consists
in sampling a large number of paths to simulate the random variables
$X^{s_0} = (X^{s_0}_n)_{n \ge 0}$; a sampling is stopped when it hits
$T$ or $\Av(T)$, and a value $1$ (resp. $0$) is assigned when $T$
(resp. $\Av(T)$) is hit; finally the average of all the values is
computed.  This requires that almost-surely a path hits $T$ or
$\Av(T)$, which is precisely decisiveness of the Markov chain
w.r.t. $T$ from $s_0$.  This allows to compute an \emph{estimate} of
the probability to reach $T$. We describe more precisely the approach
and extend the context to allow the estimation of the expected value
of $f_{L,T}$.

\begin{algorithm}[h]
  \SetKwInOut{Input}{input}\SetKwInOut{Output}{output}
  \Input{$\calC=(S,\Pt)$ a countable Markov chain, $s_0 \in S$ an
    initial state, $L : S^+ \to \bbR$ a computable function,
    $T \subseteq S$ a target set s.t. $\Av_\calC(T)$ is effective and
    $f_{L,T}$ is $B$-bounded, $\varepsilon>0$ a precision, $\delta>0$
    a confidence value.}
  $N := \Big\lceil \frac{8B^2}{\varepsilon^2}
  \log\left(\frac{2}{\delta}\right) \Big\rceil$; $\hat{f} := 0$\; \For
  {$i$ {\bf from} $1$ \textbf{{\upshape to}} $N$} { $\rho := s_0$;
    $s:=s_0$\; \lWhile{$s \notin T \cup \Av(T)$} {
      $s' := {\sf sample}(\Pt(s,\cdot))$; $\rho := \rho s'$; $s:=s'$ }
    \lIf {$s \in T$} {$\hat{f} := \hat{f} + L(\rho)$} }
  $\hat{f} := \frac{\hat{f}}{N} $; \Return
  {$[\hat{f}-\varepsilon/2, \hat{f}+\varepsilon/2]$}
  \caption{Statistical model-checking for the \EstimEV
    problem\label{algo:smc}}
\end{algorithm}

The SMC approach is presented as Algorithm~\ref{algo:smc}
(where $\Av(T)$  is assumed to be effective). This is in
general a semi-algorithm, since it may happen that the \textbf{while}
loop is never left at some iteration $i$. Nevertheless, decisiveness
ensures almost sure (a.s.) termination:

\begin{restatable}{lemma}{equivnumstat}
  \label{lemma:while}
  The \textbf{while} loop a.s. terminates if and only if
  $\calC$ is decisive w.r.t. $T$ from $s_0$.
\end{restatable}

\ifARXIV The proof of this lemma is in Appendix~\ref{app:AA}. \fi The
correctness of the algorithm will rely on this proposition that can be
straightforwardly deduced from the Hoeffding
inequality~\cite{Hoeffding63}.
\begin{proposition}
  \label{prop:hoeffding}
  Let $V_1, \ldots, V_N$ be $B$-bounded independent random variables
  and let $V = \frac{1}{N} \sum_{i=1}^N V_i$.  Let
  $\varepsilon,\delta>0$ be such that
  $N \ge \frac{8B^2}{\varepsilon^2}
  \log\left(\frac{2}{\delta}\right)$. Then:
  $\Pr\left( \left| V - \Es \big(V\big) \right| \geq
    \frac{\varepsilon}{2} \right) \leq \delta$.
\end{proposition}

We can now state the following important result.

\begin{proposition}[Termination and correctness of
  Algorithm~\ref{algo:smc}]
  \label{coro:algo:estim}
  Algorithm~\ref{algo:smc} solves the \EstimEV problem
  if and only if $\calC$ is decisive w.r.t. $T$ from $s_0$.
\end{proposition}

\begin{proof}
  The termination is a consequence of
  Lemma~\ref{lemma:while}. The correctness is a consequence of
  Proposition~\ref{prop:hoeffding}, by taking random variable $V_i$ as
  $f_{L,T}\big(\varrho^{\calC,s_0}\big)$.  In this case, $V$ is equal
  to the value of $\hat{f}$ at the end of the algorithm, which
  completes the argument.
\end{proof}

While termination is guaranteed by the previous corollary the (time)
efficiency of the simulation remains a critical factor. In particular,
the expected value $D$ of the random time $\tau^{s_0,T\cup\Av(T)}$ to
reach $T \cup \Av(T)$ from $s_0$ should be finite; in this case, the
average simulation time will be $D$ and therefore the complexity of
the whole approach will be linear in the number of
simulations. Decisiveness does not ensure this; so a dedicated
analysis needs to be done to ensure efficiency of the approach.


\section{Beyond decisiveness}
\label{sec:beyond}


In the previous section, we have presented two generic approaches for
analyzing infinite (denumerable) Markov chains. They both only apply
to \textbf{decisive} Markov chains. In this section, we twist the
previous approaches, so that they will be applicable to analyze some
\textbf{non decisive} Markov chains as well. Our proposition follows
the following steps:
\begin{itemize}
  \item based on the \textit{importance sampling} approach, we explain
        how the analysis of the original Markov chain can be transferred to
        that of a \textit{biased} Markov chain
        (Subsection~\ref{subsec:biased});
  \item we explain how a biased Markov chain can be automatically
        constructed via an \textit{abstraction}, and give conditions
        ensuring that the obtained biased Markov chain can  be
        analyzed (Subsection~\ref{subsec:abstraction});
  \item we give a generic framework based on \emph{layered Markov
          chains} and \emph{random walks}, with conditions on various
        parameters to safely apply the designed approach
        (Subsection~\ref{subsec:lmc}).
\end{itemize}

We fix an effective countable Markov chain $\calC = (S,\Pt)$,
$s_0 \in S$ an initial state, $L : S^+ \to \bbR_{\ge 0}$ a computable
function, and $T \subseteq F \subseteq S$ two effective sets, with
both $\Av_{\calC}(F)$ and $\Av_{\calC}(T)$ being effective (note that
$\Av_{\calC}(F)\subseteq \Av_{\calC}(T)$).  Since we are interested in
the probability to reach $T$, from now on, we assume that $T$ is
absorbing in $\calC$ and that $s_0\notin \Av_{\calC}(F)$.

\subsection{Model-checking via a biased Markov chain}
\label{subsec:biased}

Importance sampling has been introduced in the fifties~\cite{KH1951}
to evaluate rare-event probabilities (see the book~\cite{RT2009} for
more details). We revisit the approach in our more general setting of
reward reachability, with the extra set $F$.\footnote{The standard
  importance sampling method is recovered when $F=T$.}  The role of
$F$ will be discussed page~\pageref{roleF}.  This approach applies the
standard SMC approach with a correction factor, called
\emph{likelihood}, to another Markov chain.

\begin{definition}[biased Markov chain and likelihood]
  Let $\calC= (S,\Pt)$ with $T\subseteq F\subseteq S$ and
  $\calC' = (S',\Pt')$ be Markov chains such that:
  \begin{itemize}
    \item $S' = (S \setminus \Av_\calC(F)) \uplus \{s_-\}$, where
          $s_- \notin S$;
    \item  In $\calC'$ all states of $ T\cup\{s_-\} $ are absorbing;
    \item
          $\forall s,s' \in S \setminus \Av_\calC(F),\ \Pt(s,s')>0
            ~\Rightarrow~\Pt'(s,s')>0 \hfill
            \refstepcounter{equation}(\theequation)\label{eq:restriction} $
  \end{itemize}
  Then $\calC'$ is a \emph{biased Markov chain of $(\calC,T,F)$} and
  the \emph{likelihood} $\gamma_{\calC,\calC'}$ is the non negative
  function defined for finite paths $\rho\in {S'}^+$
  s.t. $\Pr'(\rho)>0$ by:
  $\gamma_{\calC,\calC'}(\rho)\eqdef \frac{\Pr(\rho)}{\Pr'(\rho)}$ if
  $\rho$ does not visit $s_-$, and
  $\gamma_{\calC,\calC'}(\rho)\eqdef 0$ otherwise.
\end{definition}

Eqn.~(\ref{eq:restriction}) ensures that this modification cannot
remove transitions between states of $S \setminus \Av_\calC(F)$, but
it can add transitions. So, $\Av_{\calC'}(F) = \{s_-\}$.  We fix a
biased Markov chain $\calC'$ for the rest of this subsection and omit
the subscripts for the likelihood function $\gamma$.  The likelihood
can be computed greedily from the initial state: if $\rho \cdot s$ is
a finite path of $\calC'$ such that $\gamma(\rho)$ has been computed,
then $\gamma(\rho \cdot s)$ is equal to $0$ if $s =s_-$, and
$\gamma(\rho) \cdot
  \frac{\Pt\big(\last(\rho),s\big)}{\Pt'\big(\last(\rho),s\big)}$
otherwise.

\begin{example}
  Figure~\ref{fig:ex1}(c) depicts a Markov chain which is a biased
  Markov chain of Figure~\ref{fig:ex1}(a) with $T=F=\{q_0\}$ and
  $\Av_\calC(F)=\Av_\calC(T)=\{p_0\}$.
\end{example}

\begin{figure}
  \begin{center}
    \begin{tikzpicture}[y=0.8cm,loc/.style={draw,inner sep=3pt,rounded corners=6pt},
locb/.style={draw,inner sep=1pt,circle,minimum size=0.6cm},
acol/.style={color=orange!20!black},bcol/.style={color=blue!30!black},
alpha/.style={rounded corners=2mm,dashed,draw, minimum width=3.5cm,
            minimum height=0.6cm,anchor=west,fill=gray,opacity=0.3}]

\def\xa{0.8}
\def\xab{2.6}

\begin{scope}[xshift=-0.5cm]
\node[draw,rounded corners=10pt, fill=lightgray,
      minimum width=3.7cm,minimum height=5.0cm,anchor=west] at (0.0,3.4) () {}; 

\node[alpha,minimum width= 1.7cm] 
    at (1.9,0.8) (alphm0) {};
\node[alpha] 
    at (0.1,2) (alphm1) {};
\node[alpha] 
    at (0.1,3.2) (alphm2) {};
\node[alpha] 
    at (0.1,4.5) (alphm3) {};
\node[alpha] 
    at (0.1,5.7) (alphm4) {};

\node at (2,-0.0) {\textbf{(a)} $\mathcal C$};
 
\node[loc,acol] at (\xa,0.8) (0a) {$p_0$}; 
\node[loc,accepting,bcol] at (\xab,0.8) (0b) {$q_0$}; 

\node[loc,acol] at (\xa,2) (1a) {$p_1$}; 
\node[loc,bcol] at (\xab,2) (1b) {$q_1$}; 

\path[draw,->,thick,acol] (1a) -- node[right] {0.3} (0a);
\path[draw,->,thick,bcol] (1b) -- node[right] {0.2} (0b);

\path[draw,->,thick,acol] (0a) to[loop right] node[right] {1} (0a);
\path[draw,->,thick,bcol] (0b) to[loop right] node[right] {1} (0b);

\path[draw,->,thick] ([xshift=-0.3cm]1a.west) -- (1a);

\node at (\xa,2.6) {$\cdots$};
\node at (\xab,2.6) {$\cdots$};

\node[loc,acol] at (\xa,3.2) (ima) {$p_{i\!-\!1}$}; 
\node[loc,bcol] at (\xab,3.2) (imb) {$q_{i\!-\!1}$}; 

\node[loc,acol] at (\xa,4.5) (ia) {$p_i$}; 
\node[loc,bcol] at (\xab,4.5) (ib) {$q_i$}; 

\path[draw,->,thick,acol] (ia) to[bend right=15] node[left] {0.3} (ima);
\path[draw,->,thick,bcol] (ib) to[bend left=15] node[right] {0.2} (imb);

\node[loc,acol] at (\xa,5.7) (ipa) {$p_{i+1}$}; 
\node[loc,bcol] at (\xab,5.7) (ipb) {$q_{i+1}$}; 

\path[draw,->,thick,acol] (ia) to[bend left=15] node[left] {0.3} (ipa);
\path[draw,->,thick,acol] (ia) to[bend left=35] node[right,pos=0.25,yshift=-0.1cm,xshift=-0.1cm] {0.4} (ipb);
\path[draw,->,thick,bcol] (ib) to[bend right=15] node[right] {0.4} (ipb);
\path[draw,->,thick,bcol] (ib) to[bend right=35] node[left,pos=0.25,yshift=-0.1cm,xshift=0.1cm] {0.4} (ipa);

\path[draw,->,dotted,acol] (ipa) to[bend left=15] (ia);
\path[draw,->,dotted,bcol] (ipb) to[bend right=15]  (ib);
\path[draw,->,dotted,acol] (ima) to[bend left=10] (ib);
\path[draw,->,dotted,bcol] (imb) to[bend right=10] (ia);

\path[draw,->,dotted,acol] (ima) to[bend right=15] (ia);
\path[draw,->,dotted,bcol] (imb) to[bend left=15] (ib);

\node at (\xa,6.2) {$\cdots$};
\node at (\xab,6.2) {$\cdots$};

\end{scope}

\begin{scope}[xshift=4.8cm]

\def\xb{0}

\node at (\xb,-0.0) {\textbf{(b)} $\mathcal C^\bullet$};

\node[draw,rounded corners=10pt, fill=lightgray,
      minimum width=2.0cm,minimum height=5.0cm] at (0.2,3.4) () {}; 

\node at (\xb,6.2) {$\cdots$};
\node[locb] at (\xb,0.8) (b0) {$0$}; 

\path[draw,->,thick] (b0) to[loop right] node[right] {1} (b0);
\node[locb] at (\xb,2) (b1) {$1$}; 
\path[draw,->,thick] (b1) to node[right] {$0.4$} (b0);

\node[locb] at (\xb,3.2) (bim) {\scriptsize$i\!-\!1$}; 
\node[locb] at (\xb,4.5) (bi) {$i$}; 
\node[locb] at (\xb,5.7) (bip) {\scriptsize$i\!+\!1$}; 

\path[draw,->,thick] (bi) to[bend right=15] node[left] {$0.4$} (bim);
\path[draw,->,thick] (bi) to[bend right=15] node[right] {$0.6$} (bip);

\path[draw,->,dotted] (bim) to[bend right=15] (bi);
\path[draw,->,dotted] (bip) to[bend right=15] (bi);
\node at (\xb,2.6) {$\cdots$};

\end{scope}

\draw[color=gray,->, thick] (alphm0) -> node[above,pos=0.4] {$\alpha$} (b0);
\draw[color=gray,->,thick] (alphm1) -> node[above,pos=0.4] {$\alpha$} (b1);
\draw[color=gray,->, thick] (alphm2) -> node[above,pos=0.4] {$\alpha$} (bim);
\draw[color=gray,->,thick] (alphm3) -> node[above,pos=0.4] {$\alpha$} (bi);
\draw[color=gray,->, thick] (alphm4) -> node[above,pos=0.4] {$\alpha$} (bip);

\begin{scope}[xshift=6.3cm]
\node[draw,rounded corners=10pt, fill=lightgray,
      minimum width=5.1cm,minimum height=5.0cm,anchor=west] at (0.0,3.4) () {};

\def\xab2{4.0}
\def\xc{2.4}

\node at (\xc,-0.0) {\textbf{(c)} $\mathcal C'$};

\node[loc,accepting,bcol] at (\xab2,0.8) (0b) {$q_0$}; 

\node[loc] at (\xc,0.8) (0s) {$s_-$}; 

\node[loc,acol] at (\xa,2) (1a) {$p_1$}; 
\node[loc,bcol] at (\xab2,2) (1b) {$q_1$}; 

\path[draw,->,thick,bcol] (1b) -- node[right] {$\frac{18}{60}$} (0b);

\path[draw,->,thick] (0s) to[loop right] node[right] {1} (0s);
\path[draw,->,thick,bcol] (0b) to[loop right] node[right] {1} (0b);

\path[draw,->,thick] ([xshift=-0.3cm]1a.west) -- (1a);

\node at (\xa,2.6) {$\cdots$};
\node at (\xab2,2.6) {$\cdots$};

\node[loc,acol] at (\xa,3.2) (ima) {$p_{i\!-\!1}$}; 
\node[loc,bcol] at (\xab2,3.2) (imb) {$q_{i\!-\!1}$}; 

\node[loc,acol] at (\xa,4.5) (ia) {$p_i$}; 
\node[loc,bcol] at (\xab2,4.5) (ib) {$q_i$}; 

\path[draw,->,thick,acol] (ia) to[bend right=15] node[left] {$\frac{27}{60}$} (ima);
\path[draw,->,thick,bcol] (ib) to[bend left=15] node[right] {$\frac{18}{60}$} (imb);

\node[loc,acol] at (\xa,5.7) (ipa) {$p_{i+1}$}; 
\node[loc,bcol] at (\xab2,5.7) (ipb) {$q_{i+1}$}; 

\path[draw,->,thick,acol] (ia) to[bend left=15] node[left] {$\frac{12}{60}$} (ipa);
\path[draw,->,thick,acol] (ia) to[bend left=22] node[right,pos=0.20,yshift=-0.15cm] {$\frac{16}{60}$} (ipb);
\path[draw,->,thick,bcol] (ib) to[bend right=15] node[right] {$\frac{16}{60}$} (ipb);
\path[draw,->,thick,bcol] (ib) to[bend right=21] node[left,pos=0.20,yshift=-0.15cm] {$\frac{16}{60}$} (ipa);

\path[draw,->,dotted,acol] (ipa) to[bend left=15] (ia);
\path[draw,->,dotted,bcol] (ipb) to[bend right=15]  (ib);
\path[draw,->,dotted,acol] (ima) to[bend left=10] (ib);
\path[draw,->,dotted,bcol] (imb) to[bend right=10] (ia);

\path[draw,->,dotted,acol] (ima) to[bend right=15] (ia);
\path[draw,->,dotted,bcol] (imb) to[bend left=15] (ib);

\path[draw,->,thick,bcol] (ib) .. controls (\xc,4.5) and (\xc,4.5) ..
    node[right,pos=0.7] {$\frac{10}{60}$} (0s);

\path[draw,->,thick,acol] (ia) .. controls (\xc,4.5) and (\xc,4.5) .. 
    node[left,pos=0.7] {$\frac{5}{60}$} (0s);

\path[draw,->,thick,bcol] (1b) .. controls (\xc,2) and (\xc,2) ..
    node[below,pos=0.4] {$\frac{10}{60}$} (0s);

\path[draw,->,thick,acol] (1a) .. controls (\xc,2) and (\xc,2) .. 
    node[below,pos=0.4] {$\frac{32}{60}$} (0s);

\path[draw,->,dotted,bcol] (ipb) .. controls (\xc,5.7) and (\xc,5.7) .. (0s); 
\path[draw,->,dotted,acol] (ipa) .. controls (\xc,5.7) and (\xc,5.7) .. (0s);    
\path[draw,->,dotted,bcol] (imb) .. controls (\xc,3) and (\xc,3) .. (0s); 
\path[draw,->,dotted,acol] (ima) .. controls (\xc,3) and (\xc,3) .. (0s);

\node at (\xa,6.2) {$\cdots$};
\node at (\xab2,6.2) {$\cdots$};

\end{scope}

\end{tikzpicture}
    \caption{ $\mathcal{C}, \mathcal{C^\bullet},\mathcal{C'}$ are
      three Markov chains and $\alpha$ is defined by $\alpha(q_0)=0$
      and for all $n>0$, $\alpha(p_n)=\alpha(q_n)=n$. $\mathcal{C'}$
      is a biased Markov chain of
      $(\calC,\{q_0\}) \stackrel{\alpha}{\hookrightarrow}
        (\calC^\bullet,\{0\})$.}
    \label{fig:ex1}
  \end{center}
\end{figure}
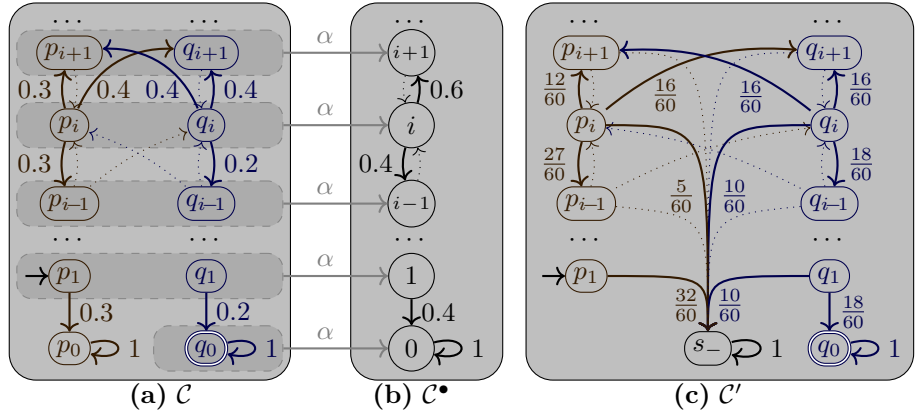

Using the likelihood, we can define the new function of interest in
Markov chain $\calC'$. We let
$L' \stackrel{\text{def}}{=} L \cdot \gamma$ and we realize that the
expected reward of $f_{L',T}$ in $\calC'$ from $s_0$ coincides with
the expected reward of $f_{L,T}$ in $\calC$ from $s_0$, as stated
below.

\begin{restatable}{proposition}{correction}
  \label{lemma:correction}
  $\Es\big(f_{L',T}\big(\varrho^{\calC',s_0}\big)\big) =
    \Es\big(f_{L,T}\big(\varrho^{\calC,s_0}\big)\big)$.
\end{restatable}

The proof of this proposition is given in Appendix~\ref{app:A}. The
idea is that the likelihood in $\calC'$ compensates for the bias in
the probabilities in $\calC'$ w.r.t. original probabilities in
$\calC$.
Thanks to this result, the computation of the expected value of
$f_{L,T}$ in $\calC$ can be reduced to the computation of the expected
value of $f_{L',T}$ in $\calC'$.  Thus, as soon as $\calC'$ and
$f_{L',T}$ satisfy the hypotheses of
Proposition~\ref{coro:algo:approx}
(resp. Proposition~\ref{coro:algo:estim}) for the \EvalEV
(resp. \EstimEV) problem, Algorithm~\ref{algo:approx}
(resp. Algorithm~\ref{algo:smc}) can be applied to $\calC'$, which
will solve the corresponding problem in $\calC$.  Specifically, the
second method is what is called the \emph{importance sampling} of
$\calC$ via $\calC'$. Observe the following facts:
\begin{itemize}[nosep]
  \item the decisiveness hypothesis only applies to the biased Markov
        chain $\calC'$, not to the original Markov chain $\calC$;
  \item the requirement that $f_{L',T}$ be $B$-bounded (for some $B$)
        does not follow from any hypothesis on $f_{L,T}$ since the
        likelihood might be unbounded.
\end{itemize}

\subsection{Construction of a biased Markov chain via an abstraction}
\label{subsec:variancegarantie}\label{subsec:abstraction}

The approach designed in the previous subsection requires the
decisiveness of the biased Markov chain and the effective
boundedness of the function which is evaluated.
We now deport these various assumptions on another Markov chain, for
which numerical (or symbolical) computations can be done, and which
will serve as an \emph{abstraction}.  This approach
generalizes~\cite{BHP12} in several directions: first,~\cite{BHP12}
was designed for finite Markov chains; then, we consider a superset
$F$ of $T$ which will allow us to relax conditions over $S\setminus T$
to its subset $S\setminus F$.

\begin{definition}
  \label{definition:abstraction}
  A Markov chain $\calC^\bullet = (S^\bullet,\Pt^\bullet)$ together
  with a set $F^\bullet$ is called an \emph{abstraction} of $\calC$
  with set $F$ by function
  $\alpha \colon S \setminus \Av_{\calC}(F) \to S^\bullet$ whenever, the
  following conditions hold:
  \begin{enumerate}[label=\emph{(\Alph*)}]
    \item for all $s \in F$, $\alpha(s) \in F^\bullet$;
    \item for all $s \in S \setminus  (F \cup \Av_{\calC}(F))$,
          ${\displaystyle \sum_{s' \notin \Av_{\calC}(F)} \Pt(s,s') \cdot
            \mu_{\calC^\bullet,F^\bullet}(\alpha(s')) \leq
            \mu_{\calC^\bullet,F^\bullet}(\alpha(s))}$.
  \end{enumerate}
\end{definition}
Condition (B) is called \emph{monotony} and is only required outside
$F \cup \Av_{\calC}(F)$.  We write more succinctly that
$(\calC^\bullet,F^\bullet)$ is an $\alpha$-abstraction of $(\calC,F)$,
denoted
$(\calC,F) \stackrel{\alpha}{\hookrightarrow}
  (\calC^\bullet,F^\bullet)$ and
$\mu_{F^\bullet} \stackrel{\text{def}}{=}
  \mu_{\calC^\bullet,F^\bullet}$ and
$\mu_F \stackrel{\text{def}}{=} \mu_{\calC,F}$.

\begin{example}
  We claim that the Markov chain $\calC^\bullet$ in
  Figure~\ref{fig:ex1}(b) with $F^\bullet=\{0\}$ is an abstraction of
  $\calC$ in Figure~\ref{fig:ex1}(a) with $s_0=p_1$. Indeed, the
  monotony condition is satisfied: for all $n>0$:
  \begin{itemize}[nosep]
    \item in $p_n$ :
          $0.3 \left(\frac{2}{3}\right)^{n+1} + 0.4
            \left(\frac{2}{3}\right)^{n+1}+0.3
            \left(\frac{2}{3}\right)^{n-1}$=
          $\frac{55}{60} \left(\frac{2}{3}\right)^{n} <
            \left(\frac{2}{3}\right)^{n}$;
    \item in $q_n$ :
          $0.4 \left(\frac{2}{3}\right)^{n+1} + 0.4
            \left(\frac{2}{3}\right)^{n+1}+0.2
            \left(\frac{2}{3}\right)^{n-1}$=
          $\frac{25}{30} \left(\frac{2}{3}\right)^{n} <
            \left(\frac{2}{3}\right)^{n}$.
  \end{itemize}
  Observe that
  $\mu_{F^\bullet}(n) = \left(\frac{0.4}{0.6}\right)^n =
    \left(\frac{2}{3}\right)^n$.
\end{example}

As will be explicit in the next lemma, an abstraction is a stochastic
bound of the initial Markov chain outside $\Av_{\calC}(F)$.

\begin{lemma}
  \label{lemma:increasereach} \label{corollary:unreachtounreach} Let
  $(\calC,F) \stackrel{\alpha}{\hookrightarrow}
    (\calC^\bullet,F^\bullet)$. Then for all
  $s \in S \setminus \Av_{\calC}(F)$,
  $\mu_F(s) \le \mu_{F^\bullet}(\alpha(s))$. In particular, for all
  $s \in S \setminus \Av_{\calC}(F)$, $\mu_{F^\bullet}(\alpha(s))>0$.
\end{lemma}

\begin{proof}
  Let $\mu_F^{(n)}(s) \eqdef \Pr(s \models \Diamond_{\le n}
    F)$. Observe that $\mu_F(s) = \lim_{n \to +\infty} \mu_F^{(n)}(s)$. We
  show by induction on $n$ that for all $s \in S$ and all $n \in \bbN$,
  $\mu_F^{(n)}(s) \le \mu_{F^\bullet}(\alpha(s))$.
  \begin{itemize}[nosep]
    \item Case $n=0$:
          \begin{itemize}[nosep]
            \item $s \in F$ implies $\alpha(s) \in F^\bullet$ (condition
                  (A)). Hence $\mu_{F^\bullet}(\alpha(s)) = 1 = \mu^{(0)}_F(s)$.
            \item $s \in S \setminus F$:
                  $\mu_F^{(0)}(s) = 0 \le \mu_{F^\bullet}(\alpha(s))$.
          \end{itemize}
    \item Inductive case:
          \begin{itemize}[nosep]
            \item $s\in F$ implies $\alpha(s) \in F^\bullet$ (condition
                  (A)). Hence $\mu_{F^\bullet}(\alpha(s)) = 1 = \mu^{(n+1)}_F(s)$.
            \item $s \in S \setminus F$:
                  \begin{footnotesize}
                    $\mu_F^{(n+1)}(s) = \sum_{s'} \Pt(s,s')
                      \cdot \mu_F^{(n)}(s') = \sum_{s' \notin \Av_{\calC}(F)} \Pt(s,s') \cdot \mu_F^{(n)}(s')
                      \le  \sum_{s' \notin \Av_{\calC}(F)} \Pt(s,s') \cdot
                      \mu_{F^\bullet}(\alpha(s'))$
                  \end{footnotesize}
                  by induction
                  hypothesis. Hence
                  $\mu_F^{(n+1)}(s) \le \mu_{F^\bullet}(\alpha(s))$ by
                  condition (B).
          \end{itemize}
  \end{itemize}
\end{proof}

Given an abstraction \label{h}
$(\calC,F) \stackrel{\alpha}{\hookrightarrow}
  (\calC^\bullet,F^\bullet)$ and $s \in S \setminus \Av_{\calC}(F)$, let
$h(s)$ be the \emph{decreasing ratio} at $s$:
${\displaystyle h(s) \stackrel{\text{def}}{=}
  \frac{1}{\mu_{F^\bullet}(\alpha(s))} \cdot \sum_{s' \in S \setminus
    \Av_{\calC}(F)} \Pt(s,s') \cdot \mu_{F^\bullet}(\alpha(s'))}$.
For all $s \in S$, $h(s)\leq 1$: this is obvious when
$s \in S \setminus F \cup \Av_{\calC}(F)$ by the monotony condition
(B); if $s \in F$, then $\alpha(s) \in F^\bullet$ by condition (A),
and hence $\mu_{F^\bullet}(\alpha(s)) = 1$.

\medskip We now define a biased Markov chain based on the above
abstraction, which will be interesting for both methods (approximation
and estimation).

\begin{definition}
  \label{def:echprefeff}
  Let
  $(\calC,F) \stackrel{\alpha}{\hookrightarrow}
    (\calC^\bullet,F^\bullet)$. Then
  $\calC' = \left(\left(S \setminus \Av_{\calC}(F)\right) \uplus
    \{s_-\},\Pt'\right)$ is the Markov chain, where $s_-$ is absorbing
  and for all $s,s' \in S \setminus \Av_{\calC}(F)$,
  ${\displaystyle \Pt'(s,s')=\Pt(s,s') \cdot
    \frac{\mu_{F^\bullet}(\alpha(s'))}{\mu_{F^\bullet}(\alpha(s))}}$
  and $\Pt'(s,s_-)=1- h(s)$.
\end{definition}
By assumption, for all $s\in T$, $s$ is absorbing in $\calC$. This
implies in particular that $s$ is also absorbing in $\calC'$.  Also,
notice that $\Pt'$ coincides with $\Pt$ within $F$, which means that
there is no bias in the zone $F$ in $\calC'$ w.r.t.  $\calC$.

\begin{lemma}
  Let
  $(\calC,F) \stackrel{\alpha}{\hookrightarrow}
    (\calC^\bullet,F^\bullet)$. Then the Markov chain $\calC'$ defined
  in Definition~\ref{def:echprefeff} is a biased Markov chain of
  $(\calC,T,F)$.
\end{lemma}

\begin{proof}
  First probabilities are well-defined, thanks to the remark on $h$
  being bounded by $1$.  The only thing which needs to be checked is
  the following: if $s,s' \notin \Av_{\calC}(F)$ and $\Pt(s,s')>0$,
  then $\Pt'(s,s')>0$. Since
  $\Pt'(s,s')=\Pt(s,s') \cdot
    \frac{\mu_{F^\bullet}(\alpha(s'))}{\mu_{F^\bullet}(\alpha(s))}$ and
  $s' \notin \Av_{\calC}(F)$ using
  Lemma~\ref{corollary:unreachtounreach},
  $\mu_{F^\bullet}(\alpha(s'))\geq \mu_{F}(s')>0$.  So $\calC'$ is a
  biased Markov chain of $(\calC,T,F)$.
\end{proof}

Since the only transitions added to $\calC$, when defining $\calC'$,
lead to $s_-$, the (qualitative) reachability of $T$ is unchanged and
so
$\Av_{\calC'}(T) = (\Av_{\calC}(T)\setminus\Av_{\calC}(F)) \cup
  \{s_-\}$.  Furthermore $\calC'$ does not depend on $T$. So we call
$\calC'$ the \emph{biased Markov chain of
  $(\calC,F) \stackrel{\alpha}{\hookrightarrow}
    (\calC^\bullet,F^\bullet)$}.  As above, we define the likelihood
$\gamma$, and accordingly the function $L' = L \cdot \gamma$. So the
approach of Subsection~\ref{subsec:biased} can be applied, provided
$\calC'$ satisfies the required properties (decisiveness and
boundedness of the evaluated function).  In
subsection~\ref{subsec:lmc}, we will be more specific and give a
generic framework guaranteeing those properties.

\medskip\noindent\textit{Role of $F$.} \label{roleF}In the original
importance sampling method, there was no superset $F \supseteq T$, and
the monotony condition was imposed on $S \setminus T$. However, in
practice, the monotony condition may not be satisfied in
$F \setminus T$ while being satisfied in $S \setminus F$; hence the
formulation with a superset $F \supseteq T$ widens the applicability
of the approach.  It should be noted that once a set $F$ has been
found, which ensures the monotony condition, any of its supersets will
also do the work. Its choice will impact the efficiency of the
approach, as will be illustrated in Section~\ref{sec:experiments},
and will therefore serve as a parameter of the approach that can be
adjusted for improving efficiency.

\medskip We end up this subsection with some property of the reward
function that is to be analyzed in the biased Markov chain obtained
using an abstraction.

\begin{restatable}{proposition}{likelihood}
  \label{prop:garantie}  Let
  $(\calC,F) \stackrel{\alpha}{\hookrightarrow}
    (\calC^\bullet,F^\bullet)$ and $L$ a computable function from $S^+$ to $\real$ such that $ f_{L,T}$ is $B$-bounded. Let $\calC'$ be the biased Markov chain of
  $(\calC,F) \stackrel{\alpha}{\hookrightarrow}
    (\calC^\bullet,F^\bullet)$ and $L'=L \cdot \gamma_{\calC,\calC'}$. Let $s_0 \in S$,
  then for every infinite path $\rho$ in $\calC'$ starting at $s_0$:
  \[
    f_{L',T}(\rho) = \left\{\begin{array}{ll}
      L(\rho_{\le
        \first_T(\rho)}) \cdot
      \mu_{F^\bullet}(\alpha(s_0)) & \text{if}\
      \rho
      \models
      \Diamond T
      \\
      0                            & \text{otherwise}
    \end{array}\right.
  \]
  Thus $f_{L',T}$ is $B$-bounded.
\end{restatable}
The proof of this proposition is given in Appendix~\ref{app:C}.  Thus
in addition to be a biased Markov chain of $\calC$, $\calC'$ preserves
a necessary condition for applying algorithms of
Section~\ref{sec:decisive}: the boundedness of the reward function.
Furthermore, when $f_{L,T} = \indic_{\Diamond T}$ (corresponding to
the standard reachability property), $f_{L',T}$ for paths starting at
$s_0$ is a bivaluated function:
$f_{L',T} = \mu_{F^\bullet}(\alpha(s_0)) \cdot \indic_{\Diamond T}$
which does not need to be computed on the fly by the algorithms.

\subsection{A generic framework based on random walks}
\label{subsec:lmc}

Our objective is to apply the algorithms of Section~\ref{sec:decisive}
to the biased Markov chain $\calC'$ defined in the previous subsection
via an abstraction, and to exploit
Proposition~\ref{prop:garantie}. This requires $\calC'$ to be
effective and to be decisive w.r.t. $T$.  The effectiveness will be
obtained via the numerical or symbolic computation (since
$\calC^\bullet$ is infinite) of $\mu_{F^\bullet}(\alpha(s))$. To that
purpose, we use random walks as abstractions since they have closed
forms for the reachability probabilities and \emph{layered Markov
  chains} as generic models.
\ifARXIV
  The proofs of this section are either omitted or sketched and full
  proofs can be found in Appendix.
\fi

\begin{definition}
  \label{def:lmc}
  A \emph{layered Markov chain} (LMC in short) is a tuple $(\calC, \lambda)$
  where $\calC = (S,\Pt)$ is a countable Markov chain, $\lambda : S  \rightarrow
    \bbN$ is a mapping such that for all $s\rightarrow_\calC s'$,
  $\lambda(s)-\lambda(s') \le 1$, and for all $n \in \mathbb{N}$,
  $\lambda^{-1}(n)$ is finite.
\end{definition}

Given $s\in S$, $\lambda(s)$ is the \emph{level} of $s$. In words
there are two requirements on $\lambda$: (1) after one step the level
can be decreased by at most one unit while it can be arbitrarily
increased, and (2) for any level $\ell$, the set of states with level
$\ell$ is finite. We define $\Pt^+(s)$, $\Pt^-(s)$ and $\Pt^=(s)$
(with $\Pt^+(s) + \Pt^-(s) + \Pt^=(s) =1$) as follows:
\begin{eqnarray*}
  \Pt^+(s)  =  \hspace*{-.5cm}\sum_{\begin{array}{c} {\scriptstyle s'
        \in S\
      \text{s.t.}} \\[-4pt]
      {\scriptstyle \lambda(s')\ge  \lambda(s)+1}\end{array}}  \hspace*{-.5cm}\Pt(s,s'),\quad
  \Pt^-(s) =  \hspace*{-.5cm}\sum_{\begin{array}{c} {\scriptstyle s'
        \in S\
      \text{s.t.}} \\[-4pt]  {\scriptstyle
      \lambda(s') = \lambda(s)-1}\end{array}}
  \hspace*{-.5cm} \Pt(s,s'),\quad
  \Pt^=(s) =  \hspace*{-.5cm}\sum_{\begin{array}{c} {\scriptstyle s'
        \in S\
      \text{s.t.}} \\[-4pt]  {\scriptstyle
      \lambda(s') = \lambda(s)}\end{array}}
  \hspace*{.5cm}\Pt(s,s')
\end{eqnarray*}

In the sequel we fix an LMC $(\calC, \lambda)$ and we consider a
finite target set $T$. We want to apply the previous approach to
$\calC$ using an $\alpha$-abstraction
$(\calC,F) \stackrel{\alpha}{\hookrightarrow}
  (\calC^\bullet,F^\bullet)$, where $\calC^\bullet$ is the random walk
$\calW^p = (\mathbb{N}, \Pt_p)$ with some probability parameter
$0<p<1$ defined as follows: $\Pt_p(0,0)=1$; for every $i >0$,
$\Pt_p(i,i+1)=p$ and $\Pt_p(i,i-1) = 1-p$ (it is depicted in
Figure~\ref{fig:ex1}(b) for $p = 0.6$). We define
$\kappa \eqdef \frac{1-p}{p}$ and recall this folk result.

\begin{restatable}{proposition}{rw}
  \label{prop:rw}
  In $\calW^{p}$, the probability to reach state $0$ from state $m$ is
  $1$ when $p\leq \frac 1 2$ and $\kappa^m$ otherwise.
\end{restatable}

Here we introduce a subclass of LMC useful for our aims.
\begin{definition}
  A LMC $(\calC, \lambda)$ is said \emph{$(p^+,N_0)$- divergent} with
  $p^+ > \frac 1 2$ and $N_0\in \bbN$ if letting
  $F\eqdef \lambda^{-1}([0,N_0])$, for every $s \in S\setminus F$,
  $\Pt^=(s)<1$ implies $\frac {\Pt^+(s)}{\Pt^-(s)+\Pt^+(s)}\geq p^+$.
\end{definition}

The $(p^+,N_0)$-divergence constrains states of levels larger than
$N_0$, and imposes that, from those states that do not stay at the
same level, the relative proportion of successors increasing their
levels compared with those decreasing their levels is at least the
value $p^+$ (itself larger than $\frac 1 2$).  Note that a
$(p^+,N_0)$-divergent LMC is also $(p'^+,N'_0)$-divergent for all
$\frac{1}{2} < p'^+ \le p^+$ and $N'_0 \ge N_0$. This will allow
to adjust the corresponding set $F$ that will be used in the
approach, as will be seen in the experiments
(Section~\ref{sec:experiments}).

To be able to apply the previous approach, it remains to examine under
which conditions starting from a $(p^+,N_0)$-\emph{divergent} LMC
$\calC$: (1) $\calW^{p}$ is an abstraction, and (2) $\calC'$ obtained
via this abstraction is decisive w.r.t. $F$ from $s_0$.  The next
proposition shows that $\calW^{p}$ is an abstraction as soon as
$1/2<p<p^+$.

\begin{restatable}{proposition}{propabs}
  \label{proposition:layered1}
  Let $(\calC, \lambda)$ be a $(p^+,N_0)$-divergent LMC and write
  $F\eqdef \lambda^{-1}([0,N_0])$.  We define $\alpha$ as the
  restriction of $\lambda$ to $S \setminus \Av(F)$, and we let
  $\frac 1 2<p< p^+$.  Then
  $(\calC,F) \stackrel{\alpha}{\hookrightarrow} (\calW^{p},[0,N_0])$
  is an abstraction.
\end{restatable}
The only point that needs to be checked is the monotony condition
defining an abstraction. The proof is given in Appendix~\ref{app:B}
and distinguishes the states that almost-surely stay within the same
level, and the other states; the rest is just calculation. The
condition which is satisfied is even stronger than monotony: for all
$s \in S\setminus (F \cup \Av(F))$ such that $\alpha(s)=n>N_0$ and
$\Pt^=(s)<1$:
$1-h(s) \ge \frac{2p-1}{(1-p)p}\cdot (\Pt^-(s)+\Pt^+(s)) \cdot
  (p^+-p)$, where $h(s)$ is the decreasing ratio at $s$, see
page~\pageref{h}.

It remains to understand under which conditions the biased Markov
chain of
$(\calC,F) \stackrel{\alpha}{\hookrightarrow} (\calW^{p},[0,N_0])$ is
decisive w.r.t. $T$. To do that, let us introduce the key notion of
\emph{attractor}~\cite{ABM07}: given a Markov chain $\calC= (S,\Pt)$
and $R\subseteq S$, $R$ is an attractor if for all $s\in S$,
$ \Pr\big(\varrho^{\calC,s} \models \Diamond R\big) =1$.  There is a
relation between attractor and decisiveness, stated as follows:
if $R$ is a \emph{finite} attractor and $B\subseteq R$, then $\calC$
is decisive w.r.t. $B$.

The next theorem gives a simple condition for a set $R$ to be an
attractor in a Markov chain, using a \emph{Lyapunov function}.

\begin{restatable}{theorem}{theoremattractor}
  \label{theorem:cnsdec}
  Let $\calC=(S,\Pt)$ be a Markov chain and $R\subseteq S$ s.t. for
  all $s\in S$,
  $\Pr\left(\varrho^{\calC,s} \models \Diamond R\right)>0$, and let
  $\mathcal L: S \rightarrow \bbR^+$ be a Lyapunov function s.t.  (1)
  for all $n\in \bbN$, $\mathcal L^{-1}([0,n])$ is finite, and (2) for
  all $s\in S\setminus R$,
  $\sum_{s'\in S}\Pt(s,s') \cdot \mathcal{L}(s')\leq \mathcal{L}(s)$.
  Then for all $s \in S$,
  $\Pr\left(\varrho^{{\calC,s}} \models \Diamond R\right)=1$.
\end{restatable}

The full proof is rather involved and partly relies on martingale
theory; it is given in Appendix~\ref{app:B}.

Using the previous theorem, we show that choosing $\calW^{p}$ as an
abstraction with $1/2<p<p^+$ ensures decisiveness of $\calC'$. The
Lyapunov function will be obtained via the level function.

\begin{restatable}{proposition}{propdecisive}
  \label{prop:layereddecisive}
  Let $(\calC, \lambda)$ be a $(p^+,N_0)$-divergent LMC, write
  $F\eqdef \lambda^{-1}([0,N_0])$, let $\alpha$ be the restriction of
  $\lambda$ to $S \setminus \Av(F)$, and fix $\frac{1}{2}<p< p^+$.
  Then the biased Markov chain $\calC'$ of
  $(\calC,F) \stackrel{\alpha}{\hookrightarrow} (\calW^{p},[0,N_0])$
  is decisive w.r.t. any $T\subseteq F$.
\end{restatable}
The detailed proof is given in Appendix~\ref{app:B}; we explain here
the rough idea. This proposition will be an application of
Theorem~\ref{theorem:cnsdec} to $\calC'$ with Lyapunov function
$\calL$ given by $\alpha$ (and additionally $\calL(s_-)=0$). So there
will be some $N_1\geq N_0$ such that
$R\eqdef \mathcal L^{-1}\left([0,N_1]\right)$ is a finite attractor in
$\calC'$. Condition (2) of Theorem~\ref{theorem:cnsdec} is ensured by
the fact that the level is unchanged after a transition from $s$ if
$\Pt^=(s)=1$, and by the stronger condition given after
Proposition~\ref{proposition:layered1} otherwise.

This proposition allows to apply the analysis of
Subsection~\ref{subsec:biased} to the biased Markov chain of
$(\calC,F) \stackrel{\alpha}{\hookrightarrow} (\calW^{p},[0,N_0])$,
yielding approximation and estimation algorithms for the original
Markov chain. Nevertheless, as argued in
Subsection~\ref{subsec:decisive_smc}, decisiveness is enough to ensure
correctness of the SMC, but not enough for efficiency. Efficiency
can be ensured, if the expected time for reaching $T \cup \Av(T)$ is
finite. We will do so by strengthening the divergence condition of LMC.

To do so we present another theorem for the existence of an attractor,
inspired by Foster's theorem~\cite{10.1214/aoms/1177728976}, whose
proof is given in Appendix~\ref{app:B}. Observe that here the
requirement becomes: the average level decreases by some fixed
$\varepsilon>0$, and the other requirements are no more necessary.

\begin{restatable}{theorem}{theoremfoster}
  \label{theo:foster-modifie}\label{foster}
  Let $\calC= (S,\Pt)$ be a Markov chain and $R \subseteq S$. If there
  exists $\calL : S \rightarrow \mathbb{R}_{\ge 0}$ and
  $\varepsilon>0$ such that for all $s \notin R$,
  $\calL(s) - \sum_{s' \in S} \Pt(s,s') \cdot \calL(s') \geq
    \varepsilon$, then for all $s \notin R$ the expected time to reach
  $R$ is finite and bounded by $\frac{\calL(s)}{\varepsilon}$; in
  particular, $R$ is an attractor of $\calC$.
\end{restatable}

We are now in a position to establish a sufficient condition for the
biased LMC $\calC'$ of
$(\calC,F) \stackrel{\alpha}{\hookrightarrow} (\calW^{p},[0,N_0])$ to
be decisive with finite expected time to reach some finite target $T$.

\begin{restatable}{proposition}{propdecisivefinite}
  \label{prop:layereddecisivefinite}
  Let $(\calC, \lambda)$ be a $(p^+,N_0)$-divergent LMC such that
  $\inf_{s \in \lambda^{-1}\left(]N_0,\infty[\right)} \Pt^+(s) >0$,
  and write $F\eqdef \lambda^{-1}([0,N_0])$.  We define $\alpha$ as
  the restriction of $\lambda$ to $S \setminus \Av(F)$, and we fix
  $\frac 1 2<p< p^+$. Then the biased Markov chain $\calC'$ of
  $(\calC,F) \stackrel{\alpha}{\hookrightarrow} (\calW^{p},[0,N_0])$
  is decisive w.r.t. $T\subseteq F$ with finite expected time to reach
  $T\cup \Av_{\calC'}(T)$.
\end{restatable}

The full proof is given in Appendix~\ref{app:B}; the idea is as
follows. We use the same Lyapunov function as before, and the stronger
condition mentioned after Proposition~\ref{prop:layereddecisive}
together with the constraint on $\Pt^+$: applying
Theorem~\ref{foster}, we are able to find a finite attractor
$R\eqdef \mathcal L^{-1}([0,N_1])\cup \{s_-\}$ for some $N_1 \ge N_0$,
reachable in finite expected time (given by $\alpha$). By analyzing
the successive visits of $R$ before reaching $T\cup \Av_{\calC'}(T)$,
we derive a bound on the expected time to reach
$T\cup \Av_{\calC'}(T)$, which (linearly) depends on the level of the
initial state.

\section{Applications and experiments}
\label{sec:experiments}


\subparagraph*{Probabilistic pushdown automata.}  Our method is
applied to the setting of probabilistic pushdown automaton (pPDA)
using the height of the stack as the level function $\lambda$.  We
only provide an informal definition for pPDA (see~\cite{Esparza06} for
a formal definition).

A pPDA configuration consists of a stack of letters from an alphabet
$\Sigma$ and a state of an automaton. A set of rules describes how the
top of the stack is modified. A rule $(q,a) \xrightarrow{w} (q',u)$
applies if the top of the stack matches the letter $a$ and the current
state is $q$. Then it replaces $a$ by the word $u$ and $q$ by $q'$.
The weight $w$ of the rule is a polynomial in $n$, the size of the
stack. Probability rules are defined with the relative weight of the
rule which applies w.r.t. all rules that could apply.  If the target
$T$ is defined as a regular language on the stack $\Av(T)$ is also a
regular language (see~\cite{Esparza06}) that can be computed: the
membership of a configuration to $T$ and $\Av(T)$ is effective and not
costly.

\begin{example}
  We consider the pPDA with a single (omitted) state with stack
  alphabet $\{\texttt{A},\texttt{B},\texttt{C}\}$ defined by the set
  of rules: \label{example:pta}
  $ \{ \texttt{A} \xrightarrow{1} \texttt{C}, \texttt{A}
  \xrightarrow{n} \texttt{BB}, \texttt{B} \xrightarrow{5} \varepsilon
  , \texttt{B} \xrightarrow{n} \texttt{AA} , \texttt{C}
  \xrightarrow{1} \texttt{C} \}$.  Starting with the stack containing
  only $\texttt{A}$, the target set $T = \{\varepsilon\}$ is the
  configuration with the empty stack and $\Av(T)$ is the set of
  configurations containing a $\texttt{C}$.  Let us describe some
  possible evolutions. From the initial configuration two rules apply
  by reading $\texttt{A}$: the new stack is $\texttt{C}$ with
  probability $\frac 1 2$ or $\texttt{BB}$ with probability
  $\frac 1 2$. From the stack $\texttt{BB}$, two rules apply by
  reading the first $\texttt{B}$: the new stack is then $\texttt{B}$
  with probability $\frac 5 7$ ($7$ is the sum of the weight of
  $\texttt{B} \xrightarrow{5} \varepsilon$ and of
  $\texttt{B} \xrightarrow{n} \texttt{AA}$, with $n=2$), and
  $\texttt{BAA}$ with probability $\frac 2 7$.
\end{example}

The approach described previously applies to pPDA, as soon as the LMC
defined by the pPDA can be proven to be $(p^+,N_0)$-divergent for some
$p^+>\frac 1 2$ and $N_0 \in \bbN$. This condition can be ensured by
some syntactical constraints on the pPDA.



\subparagraph*{Implementation.}
\label{subsec:impl}
Since SMC with importance sampling is already present in the tool
Cosmos~\cite{BHP12}, we only added the mapping function $\lambda$ in
order to apply our method.
We focus here on the
implementation details of Algorithm~\ref{algo:approx}, which
(to the best of our knowledge) has never been done.

Algorithm~\ref{algo:approx} requires to sum up a large number of
probabilities accurately while those probabilities are of different
magnitudes.  We have experimentally observed that without dedicated
summation algorithms, the implementation of this algorithm does
not converge.  We therefore propose a data structure with better
accuracy when summing up positive values at the cost of increased
memory consumption and time. This data structure encodes a floating
point number $r$ as a table of integers of size $512$ where the cell
$c$ at index $i$ stores the value $c2^{-i}$, with $c$ being a small
enough integer to be represented exactly. The probability $r$ is the
sum of the values of the table.

We specialized Algorithm~\ref{algo:approx} (called AlgoDec in the
following), when the function to be evaluated satisfies the following
\emph{monoidal property}: for all $\rho = \rho_1\rho_2$, $f(\rho) =
f(\rho_1) \cdot f(\rho_2)$; this property is in particular satisfied
by the likelihood related to an importance sampling.  It is thus
possible to merge paths leading to the same state and store only for
each state the probability to reach it and the weighted average
likelihood of the merged paths. In practice, this leads to a large
improvement. Another improvement is the use of a heap where states are
ordered by their probability to be reached: the algorithm will
converge faster. The termination of the algorithm still holds
\ifARXIV
as the
heap management is fair, see Appendix~\ref{app:termination}.
\fi


\subparagraph*{Experimental studies.}
\label{subsec:exp}
We first ran experiments\footnote{All the experiments are run with a
  timeout of 1 hour and a confidence level set to $0.99$.} on the
example depicted on Figure~\ref{fig:ex1}. As there are only two states per
level, the numerical algorithm (AlgoDec) with important sampling is
very efficient and computes an interval of $0.0258657 \pm 10^{-8}$ in
10ms. The SMC approach computes a confidence interval of
$0.02586 \pm 10^{-4}$ in 135s.  As expected the SMC approach is much
slower on such a small toy example.

The pPDA of Example~\ref{example:pta} is both decisive and a
$(p,N_0)$-divergent LMC for $1/2 <p \leq \frac{N_0}{N_0+5}$ so that
$(\mathcal{W}^p,[0,N_0])$ defines an abstraction.  We compare the use
of importance sampling with different values of $p$ to standard SMC
and AlgoDec. In Figure~\ref{fig:decif} each point is the result of a
computation with or without importance sampling. The value $0.3151$ is
contained in the intervals returned by all numerical computations and
all but one confidence intervals of SMC (consistent with 120
experiments and a confidence level of $0.99$).

Figure~\ref{fig:decifepsilon} depicts the computation time w.r.t.  the
width of the confidence interval for the two algorithms over three
Markov chains: the initial Markov chain, the importance sampling using
$\mathcal{W}^{0.6}$ as abstraction and the importance sampling with
$\mathcal{W}^{0.51}$ as abstraction. Looking only at SMC (dotted line
on the figure) the computation time scales the same way on the three
curves with the standard SMC taking more time. Looking at the AlgoDec
curves (solid line) with a well-chosen value of $p=0.6$ this algorithm
is very fast but with another value of $p$ or without importance
sampling the performance quickly degrades.

To better understand how the computation time increases w.r.t. $p$ we
plot it in Figure~\ref{fig:decifp}. The SMC is barely sensible to the
value of $p$ while the computation time of AlgoDec reaches a minimum
at around $p=0.6$ and becomes intractable when $p$ moves away from this
value.

\begin{figure}
  \begin{subfigure}[t]{0.5\textwidth}
    \includegraphics[width=1.0\textwidth]{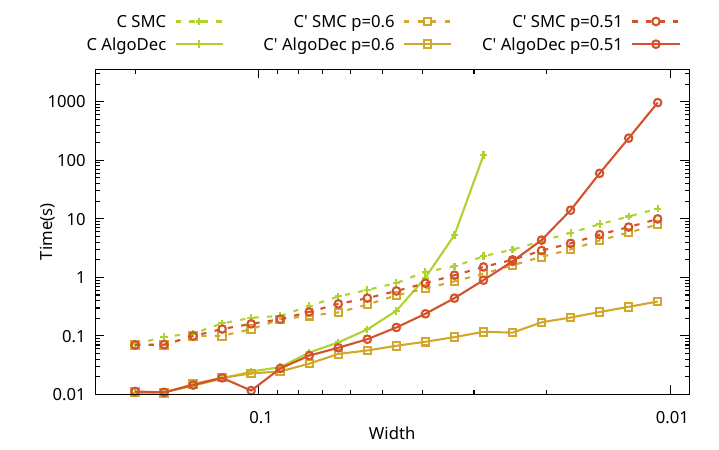}
    \caption{Computation time as a function of the precision, the
      width is given in logarithmic scale. \label{fig:decifepsilon}}
  \end{subfigure}
  \begin{subfigure}[t]{0.5\textwidth}
    \includegraphics[width=1.0\textwidth]{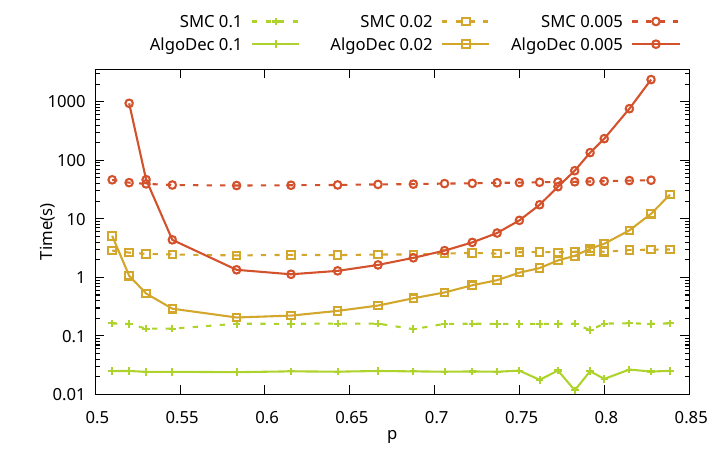}
    \caption{Computation time as a function of $p$ \label{fig:decifp}}
  \end{subfigure}
  \caption{Computation time for Example~\ref{example:pta} in logarithmic scale.  Given a value for $p$, the threshold $N_0$ is chosen as  the smallest integer such that $(\mathcal{W}^p,[0,N_0])$ defines an $\alpha$-abstraction.
    \label{fig:decif}}
\end{figure}

\begin{example}
  \label{example:pta7}
  We consider the pPDA with a single state with stack alphabet
  $\{\texttt{A},\texttt{B},\texttt{C}\}$ defined by the set of rules:
  $\{\texttt{A} \xrightarrow{1} \texttt{B}, \texttt{A} \xrightarrow{1}
  \texttt{C}, \texttt{B} \xrightarrow{10} \varepsilon, \texttt{B}
  \xrightarrow{10+n} \texttt{AA}, \texttt{C} \xrightarrow{10}
  \texttt{A}, \texttt{C} \xrightarrow{10+n} \texttt{BB}\}$
  starting with stack $\texttt{A}$, target configuration
  $T= \{\varepsilon\}$ and $\Av(T)= \emptyset$.
\end{example}

Example~\ref{example:pta7} is not decisive but is a
$(p,N_0)$-divergent LMC for $1/2<p \leq \frac{10+N_0}{20+N_0}$ thus
$(\mathcal{W}^p,[0,N_0])$ defines an abstraction. In
Figure~\ref{fig:decifepsilon2} we plot the computation time
w.r.t. $p$. The probability $0.516318$ is contained in all the
results. As in Example~\ref{example:pta}, AlgoDec is very sensitive to
the value of $p$ while SMC is not. In this example SMC is always
faster than AlgoDec with similar computation times for a well-chosen
value of $p$.

\begin{figure}
  \begin{center}
    \includegraphics[width=0.8\textwidth]{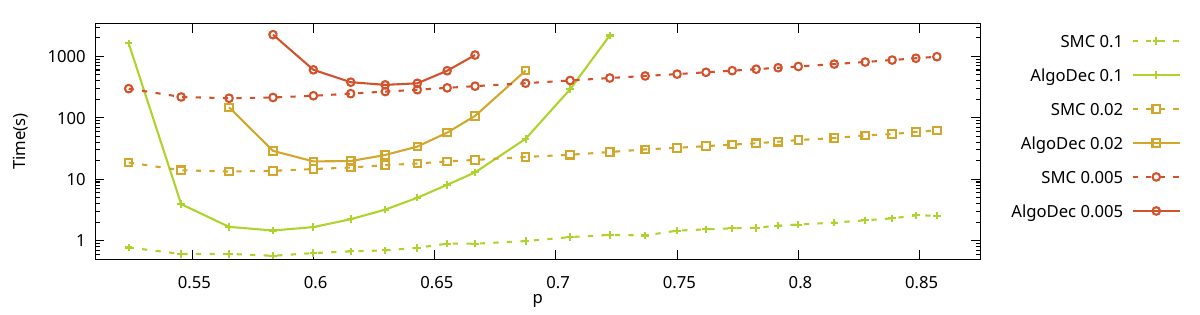}
    \vspace{-5mm}
  \end{center}
  \caption{Computation time as a function of $p$ for Example~\ref{example:pta7}.\label{fig:decifepsilon2}}
\end{figure}

From our experiments we observe that while importance sampling can be
applied both to AlgoDec and SMC, as soon as the size of the state
space grows, AlgoDec is not tractable.

Additionally, the few experiments that we have conducted suggest the
following methodology to analyze Markov chains: apply SMC with
importance sampling  for various values of $p$; find
the ``best'' $p$; apply AlgoDec with that value of $p$ (when possible).


\section{Conclusion}
\label{sec:conclusion}


We have recalled two standard approaches to the analysis of
reachability properties in infinite Markov chains, a deterministic
approximation algorithm, and a probabilistic algorithm based on
statistical model checking. For their correctness or termination, they
both require the Markov chain to satisfy a \emph{decisiveness}
property. Analyzing non decisive Markov chains is therefore a
challenge.

In this work, we have introduced the notion of abstraction for a
Markov chain and developed a theoretical method based on importance
sampling to ``transform'' a non decisive Markov chain into a decisive
one, allowing to transfer the analysis of the non decisive Markov
chain to the decisive one. Then we have presented a concrete framework
where the Markov chain is a layered Markov chain (LMC), the
abstraction is done via a random walk, and given conditions that
ensure that this abstract chain is decisive. Finally we have
implemented the two algorithms within the tool Cosmos, and compared
their respective performances on some examples given as probabilistic
pushdown automata (which are specific LMCs).

There are several further research directions that could be
investigated. First while (one-dimensional) random walks have closed
forms for reachability probabilities, other (more complex) models also
enjoy such a property, and could therefore be used for abstractions.
Second, the divergence requirements are based on conditions for
one-step transitions and could be relaxed to an arbitrary (but fixed)
number of steps. Finally, more systematic, and even automatic,
approaches could be investigated, that would compute adequate
abstractions to adequate classes of Markov chains allowing to use our
approach.

\bibliographystyle{plainurl}
\bibliography{main}

\newpage
\appendix


\ifARXIV
  \section{Missing proofs of Sections~\ref{sec:prelim} and~\ref{sec:decisive}}
  \label{app:termination}

  We give an example of a Markov chain, for which the probability to
  reach some target is irrational.  Let us consider the Markov chain
  whose set of states is $\nat$, $0$ is an absorbing state and for all
  $n>0$, $1-\Pt(n,n+1)=\Pt(n,0)=\frac 1 {n(n+1)}$.  Then the probability
  to reach 0 from 1 is equal to
  $1-\prod_{n\geq 1}1-\frac 1 {n(n+1)}=1+\frac {\cos(\sqrt{5}\pi/2)}
    {\pi}$.

  \algodec*
  \begin{proof}
    Let $\text{Tr}$ be the (possibly infinite) computation tree of the
    Markov chain, and for every depth $d$, let $\text{Tr}_{\le d}$ the
    prefix of $\text{Tr}$ of depth $d$ (it is a finite tree since the
    number of successors of each state is finite). We define
    $p_{\text{succ}}^{(d)}$ (resp. $p_{\text{fail}}^{(d)}$) the sum
    of path probabilities of successful (resp. lost) paths of length at
    most $d$. The Markov chain $\calC$ is decisive w.r.t. $T$ from $s_0$
    if and only if
    $\lim_{d \to \infty} p_{\text{succ}}^{(d)} +
      p_{\text{fail}}^{(d)} = 1$.

    \begin{itemize}
      \item Assume that $\calC$ is not decisive w.r.t. $T$ from $s_0$, and
            fix $\varepsilon>0$ such that
            $\frac{\varepsilon}{2B} < 1- \lim_{d \to \infty}
              p_{\text{succ}}^{(d)} + p_{\text{fail}}^{(d)}$. This implies
            that $\text{Tr}$ will be entirely visited. Since it is potentially infinite,
            one concludes that the algorithm will not terminate in this case.
      \item Assume that $\calC$ is decisive w.r.t. $T$ from $s_0$.  Let
            $d_\varepsilon$ be such that
            $p_{\text{succ}}^{(d_\varepsilon)} +
              p_{\text{fail}}^{(d_\varepsilon)} \ge 1-\frac{\varepsilon}{2B}$.
            Towards a contradiction, assume that the algorithm does not
            terminate. Due to the fair extraction, there is a round
            $r_\varepsilon$ of the while loop of the algorithm such that all
            vertices of $\text{Tr}_{\le d_\varepsilon}$ have been
            visited. This implies that
            $p_{\text{succ}}^{(d_\varepsilon)} +
              p_{\text{fail}}^{(d_\varepsilon)} \ge 1-\frac{\varepsilon}{2B}$,
            which contradicts the test of the while loop, and therefore the
            fact that this round has been executed.

            The set of
            infinite paths can be partitioned into three categories: (1) the ones
            whose explored prefix entering the $(n+1)^{\text{th}}$ iteration has
            reached $T$, whose set is denoted $R_n$; (2) the ones whose explored
            prefix entering the $(n+1)^{\text{th}}$ iteration has reached
            $\Av(T)$, whose set is denoted $R^-_n$; and (3) the others. Define $p_n^-$ the probability
            of the first kind of paths which corresponds to the value of $p_{\text{succ}}$
            when entering the  $(n+1)^{\text{th}}$ iteration and $p_n^+$ the sum of the probability
            of the first and third kinds of paths which corresponds to the value of $1-p_{\text{fail}}$
            when entering the  $(n+1)^{\text{th}}$ iteration.
            Using this
            decomposition, we can write:
            $ \Es\big(f_{L,T}\big(\varrho^{\calC,s_0}\big)\big) - e_n =
              (p_n^+-p_n^-) \cdot \Es\big(f_{L,T}\big(\varrho^{\calC,s_0}\big)
              \mid \varrho^{\calC,s_0} \notin R_n \cup R^-_n) $. Thus, since
            $f_{L,T}$ is $B$-bounded:
            $ |\Es\big(f_{L,T}\big(\varrho^{\calC,s_0}\big)\big) -
              e_n|\leq(p_n^+-p_n^-) \cdot B $.  
            We deduce that $\Es\big(f_{L,T}\big(\varrho^{\calC,s_0}\big)\big)$
            belongs to the interval $[e_n-B(p_n^+-p_n^-),e_n+B(p_n^+-p_n^-)]$.
            This interval has length at most $\varepsilon$ since the loop is
            left when $|p_n^+-p_n^-| \le \varepsilon/2B$, which allows us to
            conclude.
    \end{itemize}
  \end{proof}

  \label{app:AA}

  \equivnumstat*
  \begin{proof}
    The probability of non termination of the {\bf while} loop is the
    probability that an infinite random path never meets $T\cup
      \Av(T)$. By definition, this probability is null if and only if
    $\calC$ is decisive w.r.t. $T$.
  \end{proof}

\fi

\section{Some missing proofs of Section~\ref{sec:beyond}}

\ifARXIV
  \subsection{Few elements of martingale theory}

  We recall here some results on martingales, which are useful for our
  work.

  \begin{definition}
    Let $(\Omega, \mathcal F, \Pr)$ be a probabilistic space,
    $\mathcal H\subseteq \mathcal F$ be a $\sigma$-algebra, $X$ be a
    random variable $\mathcal F$-measurable with $\Es(|X|)<\infty$. Then
    there exists a $\mathcal H$-measurable random variable (r.v.)
    $\Es\left(X \mid \mathcal H\right)$, called the conditional
    expectation of $X$ w.r.t. $\mathcal{H}$, s.t. for all
    $H\in \mathcal H$,
    $\int _{H}X d\Pr=\int _{H}\Es\left(X \mid\mathcal H\right) d\Pr$.

    Furthermore for all $\mathcal H$-measurable r.v. $Y$ satisfying the
    condition $\int _{H}Xd\Pr=\int _{H}Yd\Pr$ for all $H\in \mathcal H$,
    one has $\Pr\left(Y\neq \Es(X \mid \mathcal H)\right)=0$.
  \end{definition}

  \begin{definition}
    A \emph{filtered space}
    $(\Omega,\mathcal F, (\mathcal F_n)_{n\in \bbN}, \Pr)$ is defined
    by:
    \begin{itemize}[nosep]
      \item $(\Omega, \mathcal F, \Pr)$ be a probabilistic space;
      \item $(\mathcal F_n)_{n\in \bbN}$, a sequence of $\sigma$-algebras
            s.t. for all $n\in \bbN$,
            $\mathcal F_n\subseteq \mathcal F_{n+1}\subseteq \mathcal F$.
    \end{itemize}
    The sequence $(\mathcal F_n)_{n\in \bbN}$ is called a
    \emph{filtration}.
  \end{definition}

  A sequence ${\mathbf X}= (X_n)_{n\in \bbN}$ of random variables is
  called a \emph{process}.
  \begin{definition}
    Let $(\Omega,\mathcal F, (\mathcal F_n)_{n\in \bbN}, \Pr)$ be a
    filtered space and $ {\mathbf X}=(X_n)_{n\in \bbN}$ be a process.
    Then $ {\mathbf X}$ is \emph{adapted} to
    $\left(\mathcal F_n\right)_{n\in \bbN}$ if for all $n\in \bbN$,
    $X_n$ is $\mathcal F_n$-measurable.  If furthermore for all
    $n\in \bbN$, $\Es\left(|X_n|\right)<\infty$ and
    $\Es\left(X_{n+1}\mid \mathcal F_n \right)\leq X_{n}$ a.s. then
    ${\mathbf X}$ is a \emph{supermartingale}.
  \end{definition}

  Note that there is always a filtration $(\mathcal F_n)_{n\in \bbN}$
  such that $ {\mathbf X}$ is \emph{adapted} to
  $(\mathcal F_n)_{n\in \bbN}$, for instance defining $\mathcal{F}_n$ as
  the smallest $\sigma$-algebra such that all $X_i$ with $i \le n$ are
  measurable.

  \begin{proposition}
    Let $(\Omega,\mathcal F, (\mathcal F_n)_{n\in \bbN}, \Pr)$ be a
    filtered space, $(X_n)_{n \in \bbN}$ be a non negative
    supermartingale. Then
    $X_\infty\eqdef \lim_{n \rightarrow \infty}X_n$ exists almost surely
    and $\Es(X_\infty)\leq \Es(X_0)$.
  \end{proposition}
\fi

\subsection{Proofs of results of Section~\ref{subsec:biased}}
\label{app:A}

\ifARXIV
  Before going to the proof of Proposition~\ref{lemma:correction}, we
  state a useful lemma.
\fi

\begin{lemma}
  \label{lemma2}
  $\Pr(\rho) > 0$ and  $\last(\rho)\notin \Av_{\calC}(F)$ imply
  $\Pr'(\rho)>0$.
\end{lemma}

\ifARXIV
  \begin{proof}
    The proof is done by induction.  When $\rho$ is reduced to
    $s_0$, we immediately get $\Pr(\rho) = \Pr'(\rho)=1$.
    Assume now that $\rho=\rho' s$ with $\Pr(\rho) > 0$ and
    $\last(\rho)\notin \Av_{\calC}(F)$.  Then $\last(\rho')\notin  \Av_{\calC}(F)$
    and $\Pt(\last(\rho'),s)>0$.  Applying the induction hypothesis to
    $\rho'$, we get $\Pr'(\rho')>0$.  Applying the condition on $\Pt'$,
    $\Pt'(\last(\rho'),s)>0$.  Thus
    $\Pr'(\rho)=\Pr'(\rho')\Pt'(\last(\rho'),s)>0$.
  \end{proof}
\fi

\ifARXIV
  \correction*
  \begin{proof}
    We proceed by a sequence of equalities:
    \begin{small}
      \[
        \begin{array}{c@{\;}l@{\;}c@{\;}l}
                                                                                                               & \Es\big(f_{L',T}\big(\varrho^{\calC',s_0}\big)\big)                                            & = & \Es\big(\big(f_{L',T} \cdot \indic_{\neg \Diamond T}\big)\big(\varrho^{\calC',s_0}\big)
          + \big(f_{L',T} \cdot \indic_{\Diamond T} \big) \big(\varrho^{\calC',s_0}\big)\big)
          \\[5mm]
          =                                                                                                    & \Es \big(\big(f_{L',T} \cdot \indic_{\Diamond T} \big) \big(\varrho^{\calC',s_0}\big)\big)
                                                                                                               &
          =                                                                                                    & {\displaystyle \sum_{\substack{\rho \in (S\setminus \Av_\calC(F))^+\ \text{s.t.}                                                                                                             \\
                  \Pr'(\rho)>0\ \text{and}\
                  \last(\rho) = \first_T(\rho)}}  \hspace*{-1cm}L'(\rho) \cdot {\Pr}'(\rho)}
          \\[5mm]
          =                                                                                                    & {\displaystyle\hspace*{-1cm} \sum_{\substack{\rho \in (S\setminus \Av_\calC(F))^+\ \text{s.t.}                                                                                               \\
                  \Pr'(\rho)>0\ \text{and}\ \last(\rho) = \first_T(\rho)}}  \hspace*{-1cm}L(\rho)
              \cdot
              \frac{\Pr(\rho)}{\Pr'(\rho)} \cdot \Pr'(\rho)}
                                                                                                               &
          =                                                                                                    & {\displaystyle \sum_{\substack{\rho \in (S\setminus \Av_\calC(F))^+\
          \text{s.t.}                                                                                                                                                                                                                                                                                         \\ \Pr'(\rho)>0\ \text{and}\ \last(\rho) = \first_T(\rho)}}
          \hspace*{-1cm}L(\rho) \cdot \Pr(\rho)}                                                                                                                                                                                                                                                              \\[5mm]
          =                                                                                                    & {\displaystyle \hspace*{-1cm}\sum_{\substack{\rho \in (S\setminus \Av_\calC(F))^+\
          \text{s.t.}                                                                                                                                                                                                                                                                                         \\ \Pr'(\rho)>0,\ \Pr(\rho)>0\ \text{and}\ \last(\rho) = \first_T(\rho)}}
              \hspace*{-1cm}L(\rho) \cdot \Pr(\rho)}
                                                                                                               &
          =                                                                                                    & {\displaystyle \sum_{\substack{\rho \in (S\setminus \Av_\calC(F))^+\
          \text{s.t.}                                                                                                                                                                                                                                                                                         \\ \Pr(\rho)>0,\ \text{and}\ \last(\rho) = \first_T(\rho)}}
          \hspace*{-1cm}L(\rho) \cdot \Pr(\rho)}\  {\text{\tiny (by\ Lemma~\ref{lemma2})}}                                                                                                                                                                                                                    \\[5mm]
          =                                                                                                    & {\displaystyle \hspace*{-1cm}\sum_{\substack{\rho \in S^+\
          \text{s.t.}                                                                                                                                                                                                                                                                                         \\ \Pr(\rho)>0,\ \text{and}\ \last(\rho) = \first_T(\rho)}}
          \hspace*{-1.5cm}L(\rho) \cdot \Pr(\rho)}\ {\emph{\tiny(since $\Av_\calC(F)\subseteq \Av_\calC(T)$)}}  &
          =                                                                                                    & \Es\big(f_{L,T}\big(\varrho^{\calC,s_0}\big)\big)
        \end{array}
      \]
    \end{small}
  \end{proof}
\fi

\subsection{Proofs of results of Section~\ref{subsec:abstraction}}
\label{app:C}

\ifARXIV
  We target the proof of Proposition~\ref{prop:garantie}.
  Before that we
  establish properties of the biased Markov chain.
\fi

\begin{lemma}
  \label{lemma:T-2}Let $\calC'$ be the biased Markov chain of
  $(\calC,F) \stackrel{\alpha}{\hookrightarrow}
    (\calC^\bullet,F^\bullet)$. Then for all
  $s\in S\setminus \Av_{\calC}(F)$,
  \begin{itemize}[nosep]
    \item for all $\rho$ starting from $s$ with $\Pr'(\rho)>0$ and
          $\last(\rho) \ne s_-$,
          $\Pr(\rho)=\Pr'(\rho) \cdot
            \frac{\mu_{F^\bullet}(\alpha(s))}{\mu_{F^\bullet}(\alpha(last(\rho)))}$;
    \item
          $\mu_{\calC,T}(s) =\mu_{F^\bullet}(\alpha(s)) \cdot\mu_{\calC',T}(s)$ and
          $\mu_{\calC,F}(s) =\mu_{F^\bullet}(\alpha(s)) \cdot\mu_{\calC',F}(s)$.
  \end{itemize}
\end{lemma}

\ifARXIV
  \begin{proof}
    We establish the first property by induction. Let $\rho$ be a path
    starting from $s$ with $\Pr'(\rho)>0$ and $\last(\rho) \ne s_-$. If
    $\rho$ is the single state $s$ then $\Pr'(\rho)=\Pr(\rho)=1$. Since
    $last(\rho)=s$, the base case is proved.

    Let $\rho=\rho' s''$ with $\Pr'(\rho)>0$ and
    $s' \ne s_-$ with $s'$ denoting $\last(\rho')$.
    \[
      \Pr(\rho)=\Pr(\rho') \cdot \Pt(s',s'')=\Pr(\rho') \cdot
      \Pt'(s',s'') \cdot
      \frac{\mu_{F^\bullet}(\alpha(s'))}{\mu_{F^\bullet}(\alpha(s''))}.
    \]
    It is well-defined due to
    Lemma~\ref{corollary:unreachtounreach}. Applying the induction
    hypothesis,
    we get:
    \[
      \Pr(\rho)=\Pr'(\rho') \cdot
      \frac{\mu_{F^\bullet}(\alpha(s))}{\mu_{F^\bullet}(\alpha(s'))}
      \cdot \Pt'(s',s'') \cdot
      \frac{\mu_{F^\bullet}(\alpha(s'))}{\mu_{F^\bullet}(\alpha(s''))}=
      \Pr'(\rho) \cdot
      \frac{\mu_{F^\bullet}(\alpha(s))}{\mu_{F^\bullet}(\alpha(s''))}.
    \]
    This shows the induction step.

    The paths that reach $T$ (resp. $F$) from $s$ are the same in
    $\calC$ and $\calC'$, and they do not reach $s_-$.  Pick such a path
    $\rho$. Then due to the previous property:
    $\Pr(\rho)=\Pr'(\rho) \cdot \mu_{F^\bullet}(\alpha(s))$.  Summing
    over all such paths establishes the second property.
  \end{proof}
\fi

While the previous property allows to solve the reachability problem
in $\calC$ using $\calC'$, the next proposition extends it to the reward
reachability problem.

\ifARXIV
\likelihood*
\begin{proof}
\else
\begin{proof}[Proof of Proposition~\ref{prop:garantie}]
  \fi
  By definition, $f_{L',T}$ assigns $0$ to infinite paths not visiting
  $T$. Assume now that $\rho$ is an infinite path visiting $T$.  Then
  $f_{L',T}(\rho) = (L \cdot \gamma) \big(\rho_{\le
        \first_T(\rho)}\big)$. Let $\rho' = \rho_{\le \first_T(\rho)}$. It
  does not visit $s_-$, hence
  $\gamma(\rho') = \frac{\Pr(\rho')}{\Pr'(\rho')} =
    \frac{\mu_{F^\bullet}(\alpha(s))}{\mu_{F^\bullet}(\alpha(\last(\rho')))}$
  by Lemma~\ref{lemma:T-2}. Since $\last(\rho') \in T$,
  $\gamma(\rho') = \mu_{F^\bullet}(\alpha(s))$. This implies the first
  part of the proposition.  The restriction to the case of the
  indicator function is immediate.
\end{proof}

\subsection{Proofs of results of Section~\ref{subsec:lmc}}
\label{app:B}



We first establish that random walks parametrized by
$\frac 1 2 <p <p^+$ are abstractions for a $(p^+,N_0)$- divergent LMC
and give useful information on the decreasing ratio for states $s$
with $\Pt^+(s)<1$.

\ifARXIV
  \propabs*

  \begin{proof}
    We denote $\mu_{\calW^{p},[0,N_0]}$ more simply by
    $\mu^\bullet_{[0,N_0]}$.  We pick $s \in S\setminus (F \cup \Av(F))$
    such that $\alpha(s)=n>N_0$, and we distinguish between two cases.

    \begin{description}
      \item[Case $\Pt^=(s)=1$.] Observe that for all $s'$ such that
            $\Pt(s,s')>0$, $\alpha(s')=n$. Thus:
            \[
              \sum_{s' \notin \Av_{\calC}(F)} \Pt(s,s') \cdot
              \mu^\bullet_{[0,N_0]}(\alpha(s'))= \mu^\bullet_{[0,N_0]}(n)
              \sum_{s' \notin \Av_{\calC}(F)} \Pt(s,s') \leq
              \mu^\bullet_{[0,N_0]}(n)=\mu^\bullet_{[0,N_0]}(\alpha(s))
            \]
      \item[Case $\Pt^=(s)<1$.]
            We can compute:
            \begin{small}
              \[
                \begin{array}{r@{\;}cl}
                  1-h(s) & =    & {\displaystyle 1 -
                      \frac{1}{\mu^\bullet_{[0,N_0]}(\alpha(s))} \left(\sum_{s'
                          \in S \setminus \Av(F)}
                      \mu^\bullet_{[0,N_0]}(\alpha(s')) \cdot \Pt(s,s')
                  \right)}                                                                                                                       \\
                         & \ge  & {\displaystyle 1 -
                      \frac{1}{\mu^\bullet_{[0,N_0]}(n)}
                      \left(\mu^\bullet_{[0,N_0]}(n-1)\hspace*{-1cm} \sum_{s' \in S \setminus
                          \Av(F)\ \text{s.t.}\ \alpha(s') =n-1}\hspace*{-1cm}

                      \Pt(s,s') +
                      \mu^\bullet_{[0,N_0]}(n)\hspace*{-1cm}  \sum_{s' \in S \setminus
                          \Av(F)\ \text{s.t.}\ \alpha(s') =n}\hspace*{-1cm}

                      \Pt(s,s') \right.}
                  \\ & &
                   \multicolumn{1}{r}{ + \displaystyle \left.
                      \mu^\bullet_{[0,N_0]}(n+1)\hspace*{-1cm}  \sum_{s' \in S \setminus
                          \Av(F)\ \text{s.t.}\ \alpha(s') > n}\hspace*{-1cm}
                      \Pt(s,s')\right)} 
                  \\
                         &      & \text{(since $\Pt^+(s) + \Pt^-(s) + \Pt^=(s) =1$ and $\mu^\bullet_{[0,N_0]}(x)$ is non increasing w.r.t. $x$)}
                  \\[0.5cm]
                         & \geq & {\displaystyle 1 - \frac{1}{\mu^\bullet_{[0,N_0]}(n)}\left( \mu^\bullet_{[0,N_0]}(n-1)
                      \cdot \Pt^-(s)+ \mu^\bullet_{[0,N_0]}(n)
                      \cdot \Pt^=(s)  +
                      \mu^\bullet_{[0,N_0]}(n+1) \cdot
                  \Pt^+(s)\right)}                                                                                                               \\
                         & =    & {\displaystyle 1 -\Pt^=(s)- \frac{1}{\mu^\bullet_{[0,N_0]}(n)}\left( \mu^\bullet_{[0,N_0]}(n-1)
                      \cdot \Pt^-(s) +
                      \mu^\bullet_{[0,N_0]}(n+1) \cdot
                  \Pt^+(s)\right)}                                                                                                               \\
                         & =    & {\displaystyle \Pt^+(s) +\Pt^-(s)- \frac{1}{\mu^\bullet_{[0,N_0]}(n)}\left( \mu^\bullet_{[0,N_0]}(n-1)
                      \cdot \Pt^-(s) +
                      \mu^\bullet_{[0,N_0]}(n+1) \cdot
                  \Pt^+(s)\right)}                                                                                                               \\
                         & =    & {\displaystyle  \frac{(\mu^\bullet_{[0,N_0]}(n)- \mu^\bullet_{[0,N_0]}(n-1))
                        \cdot \Pt^-(s) +
                        (\mu^\bullet_{[0,N_0]}(n)-  \mu^\bullet_{[0,N_0]}(n+1)) \cdot
                  \Pt^+(s)}{\mu^\bullet_{[0,N_0]}(n)}}                                                                                           \\
                         & =    & {\displaystyle  (1- \frac{1}{\kappa})
                      \cdot \Pt^-(s) +
                      (1-  \kappa) \cdot
                  \Pt^+(s)}                                                                                                                      \\
                         & =    & {\displaystyle  ( \Pt^-(s)+\Pt^+(s))\left((1- \frac{1}{\kappa})
                      \cdot \frac{\Pt^-(s)}{\Pt^-(s)+\Pt^+(s)} +
                      (1-  \kappa) \cdot
                  \frac{\Pt^+(s)}{\Pt^-(s)+\Pt^+(s)}\right) }                                                                                    \\
                         & =    & {\displaystyle  ( \Pt^-(s)+\Pt^+(s))\left((1- \frac{1}{\kappa})
                      \cdot(1- \frac{\Pt^+(s)}{\Pt^-(s)+\Pt^+(s)} )+
                      (1-  \kappa) \cdot
                  \frac{\Pt^+(s)}{\Pt^-(s)+\Pt^+(s)}\right) }                                                                                    \\
                         & =    & {\displaystyle (2p-1)( \Pt^-(s)+\Pt^+(s))\left( -\frac{1}{1-p}
                      \cdot (1- \frac{\Pt^+(s)}{\Pt^-(s)+\Pt^+(s)} )+
                      \frac{1}{p} \cdot
                  \frac{\Pt^+(s)}{\Pt^-(s)+\Pt^+(s)}\right) }                                                                                    \\
                         & =    & \frac{2p-1}{(1-p)p} {\displaystyle ( \Pt^-(s)+\Pt^+(s))\left( -p
                    \cdot (1- \frac{\Pt^+(s)}{\Pt^-(s)+\Pt^+(s)} )+
                    (1-p) \cdot
                  \frac{\Pt^+(s)}{\Pt^-(s)+\Pt^+(s)}\right) }                                                                                    \\
                         & =    & \frac{2p-1}{(1-p)p} {\displaystyle ( \Pt^-(s)+\Pt^+(s))\left( \frac{\Pt^+(s)}{\Pt^-(s)+\Pt^+(s)} -p
                  \right) }                                                                                                                      \\
                         & =    & \frac{2p-1}{(1-p)p} {\displaystyle ( \Pt^-(s)+\Pt^+(s))\left( -p
                    \cdot (1- \frac{\Pt^+(s)}{\Pt^-(s)+\Pt^+(s)} )+
                    (1-p) \cdot
                  \frac{\Pt^+(s)}{\Pt^-(s)+\Pt^+(s)}\right) }                                                                                    \\
                         & \geq & \frac{2p-1}{(1-p)p} {\displaystyle ( \Pt^-(s)+\Pt^+(s))\left(p^+-p
                    \right) > 0}
                  \\
                         &      & \text{(since $\frac{1}{2} < p < p^+$)}
                \end{array}
              \]
            \end{small}
    \end{description}
    This concludes the proof of monotony, implying that
    $(\calC,F) \stackrel{\alpha}{\hookrightarrow} (\calW^{p},[0,N_0])$
    is an abstraction.
  \end{proof}

\fi

Out of the above proof, a stronger condition than monotony happens to
be satisfied.

\begin{corollary}
  \label{coro:stronger}
  Let $(\calC, \lambda)$ be a $(p^+,N_0)$-divergent LMC.  We
  define $\alpha$ as the restriction of $\lambda$ to
  $S \setminus \Av(F)$, and we let $\frac 1 2<p< p^+$.  Then the
  monotony condition satisfied by the abstraction
  $(\calC,F) \stackrel{\alpha}{\hookrightarrow} (\calW^{p},[0,N_0])$
  can be strengthened as follows. For all $s \in S\setminus (F \cup \Av(F))$
  such that $\alpha(s)=n>N_0$ and $\Pt^=(s)<1$:
  $1-h(s) \ge \frac{2p-1}{(1-p)p}\cdot (\Pt^-(s)+\Pt^+(s)) \cdot
    (p^+-p)> 0$, where $h(s)$ is the decreasing ratio at $s$, see
  page~\pageref{h}.
\end{corollary}

Before turning to the proof of Theorem~\ref{theorem:cnsdec}, we first
establish a sufficient condition to be an attractor in a Markov chain.

\begin{lemma}
  \label{lemma:csdec}
  Let $\calC=(S,\Pt)$ be a Markov chain, $s_0\in S$ and $R\subseteq S$
  s.t. for all $s\in S$,
  $\Pr\left(\varrho^{\calC,s} \models \Diamond R\right)>0$.  Assume
  that for every $\delta>0$, there exists a finite set
  $S_\delta\subseteq S$ and $m_\delta \in \bbN$ such that
  $\Pr\left(\bigwedge_{i \ge m_\delta} X^{\calC, s_0}_i\in
    S_\delta\right)>1-\delta$.  Then
  $\Pr\left(\varrho^{\calC,s_0} \models \Diamond R\right)=1$.
\end{lemma}
\begin{proof}
  Fix some $\delta>0$.  Let $\ell \in \bbN$ be the maximal length over
  $s \in S_\delta$ of a shortest path from $s$ to $R$ and $p_{\min}>0$
  the minimal probability of these paths.  Let $k\in \bbN$. Then
  $\Pr\left( \bigwedge_{m_\delta\leq i \leq m_\delta+k\ell}
    X^{\calC,s_0}_i\notin R \mid \bigwedge_{m_\delta\leq i}
    X^{\calC,s_0}_i\in S_\delta\right)\leq (1-p_{\min})^k$.

  Letting
  $k$ go to $\infty$, one gets
  $\Pr\left(\bigwedge_{m_\delta\leq i} X^{\calC,s_0}_i\notin R \mid
    \bigwedge_{ m_\delta\leq i} X_i\in S_\delta\right)=0$ implying
  $\Pr\left(\varrho^{\calC,s_0} \models \Diamond R \mid
    \bigwedge_{m_\delta\leq i} X^{\calC,s_0}_i\in S_\delta\right)=1$.
  Thus
  $\Pr\left(\varrho^{\calC,s_0} \models \Diamond
    R\right)>1-\delta$. Since $\delta$ is arbitrary, one gets
  $\Pr\left(\varrho^{\calC,s_0} \models \Diamond R\right)=1$.
\end{proof}

Then using both martingale theory and the previous lemma we establish
another sufficient condition based on a non negative state function
non increasing on average (i.e. the expected next value).  A similar
proof for recurrence of irreducible Markov chains can be found
in~\cite{FMM1995}.

\ifARXIV
\theoremattractor*
\begin{proof}
\else
\begin{proof}[Proof of Theorem~\ref{theorem:cnsdec}]
  \fi
  Since we are interested in reachability of $R$, w.l.o.g. we assume
  that all $s\in R$ are absorbing and thus
  $\sum_{s'\in S}\Pt(s,s')\cdot \mathcal L(s')= \mathcal L(s)$.

  We fix some initial state $s_0$ and consider the random sequence of
  states $(X^{\calC,s_0}_n)_{n\in \bbN}$, which we simply write
  $(X_n)_{n \in \bbN}$.  Define $\mathcal F_n$ the $\sigma$-algebra
  generated by $(X_m)_{m\leq n}$ and $Y_n= \mathcal L(X_n)$.  Due to
  the inequation
  $\sum_{s'\in S}\Pt(s,s') \mathcal L(s')\leq \mathcal L(s)$ for all
  $s \in S$ and the memoryless property of Markov chain
  $\Es\left(Y_{n+1}\mid \mathcal F_n\right)=\Es\left(Y_{n+1}\mid
    X_n\right)= \sum_{s'\in S}\Pt(X_n,s') \cdot \calL (s') \leq
    \calL(X_n)=Y_n$. Thus $(Y_n)_{n\in \bbN}$ is a
  supermartingale. Consider the limit $Y_\infty$ of this supermartingale: it satisfies
  $\Es\left(Y_\infty\right) \leq \mathcal L(s_0)<\infty$.

  We use Lemma~\ref{lemma:csdec} to conclude that $R$ is an
  attractor. Towards a contradiction assume that the sufficient
  condition of Lemma~\ref{lemma:csdec} is not satisfied. There is some
  $\delta>0$ such that for all finite set $S^*$ and $m\in \bbN$,
  $\Pr\left(\bigvee_{m\leq i} X_i\notin S^*\right)\geq \delta$.  For
  every $n \in \bbN$, choose $S^*_n= \mathcal L^{-1}([0,n])$. The
  event
  $E_1=\left\{\bigwedge_{m\in \bbN}\bigvee_{m\leq i} X_i\notin
    S^*_n\right\}$ is the limit of decreasing events with probability
  larger than or equal to $\delta$. So $\Pr(E_1)\geq \delta$. Consider
  the event $E_2=\{Y_\infty\geq n\}$: then $E_1\subseteq E_2$. Thus
  $\Pr\left(E_2\right)\geq \delta$.  Now, by Markov's inequality
  applied to the random variable $Y_\infty$,
  $\Es\left(Y_\infty\right)\geq n \cdot \Pr\left(Y_\infty\geq
    n\right)\geq n\delta$. Since this is true for all $n$,
  $\Es\left(Y_\infty\right)=\infty$, which is a contradiction. The
  sufficient condition of Lemma~\ref{lemma:csdec} is then satisfied,
  which implies that $R$ is an attractor.
\end{proof}

The next proposition shows that choosing $\calW^{p}$ as an abstraction
with $1/2<p<p^+$ ensures decisiveness of $\calC'$.

\ifARXIV
  \propdecisive*
  \begin{proof}
    From Proposition~\ref{proposition:layered1}, $(\calW^p,[0,N_0])$
    is a $\alpha$-abstraction of $(\calC,F)$, so $\calC'$ is
    well-defined.  We exhibit some $N_1 \ge N_0$ s.t.
    $\alpha^{-1}([0,N_1]) \cup \{s_-\}$ is a finite attractor of
    $\calC'$, which implies decisiveness of $\calC'$ w.r.t. $T$
    (thanks to Lemma~{3.4} of~\cite{ABM07}).

    To do so, we apply Theorem~\ref{theorem:cnsdec} to the Markov chain
    $\calC$, using the layered function $\calL$, which coincides with
    $\alpha$ on $S \setminus \Av(F)$ and extended by $\calL(s_-)=0$ as
    the Lyapunov function. It remains to show the inequation on
    $\mathcal{L}$.

    Let $s \in S'\setminus \{s_-\}$ with $\alpha(s)=n>N_0$.
    \begin{description}
      \item[Case $\Pt^=(s)=1$.] We compute:
            \[
              \begin{array}{rcl}
                \sum_{s' \in S'} \Pt'(s,s') \cdot \calL(s') & =   & \sum_{s' \in
                  S\setminus
                  \Av(F)}
                \Pt'(s,s')
                \cdot
                \calL(s')                                                                                                     \\
                                                            & =   & n \cdot \sum_{s' \in S\setminus \Av(F)} \Pt'(s,s')        \\
                                                            &     & { \text{\small(since for every $s' \in S$, $\Pt'(s,s')>0$
                implies $\alpha(s')=\alpha(s)$)}}                                                                             \\
                                                            & \le & n = \mathcal{L}(s)
              \end{array}
            \]
      \item[Case $\Pt^=(s)<1$.] We compute:
            \begin{flalign*}
                 & \calL(s) - \sum_{s' \in S'}\Pt'(s,s') \cdot \calL(s')        \\
              =~ & \sum_{s' \in S'}\Pt' (s,s') \cdot (\calL(s) - \calL(s'))     \\
              =~ & \sum_{k \ge n+1} \sum_{s' \in S'\ \text{s.t.}\ \calL(s') =k}
              (n-k)\Pt' (s,s') + \sum_{s' \in S'\ \text{s.t.}\ \calL(s') = n-1}
              \Pt' (s,s') +
              n\Pt' (s,s_-)                                                     \\
              =~ & \sum_{k \ge n+1} \;\sum_{s' \in S'\ \text{s.t.}\ \calL(s') =k} \hspace{-0.3cm}
              - \kappa^{k-n} (k-n)\Pt (s,s') + \hspace{-0.5cm} \sum_{s' \in S'\ \text{s.t.}\
                \calL(s') = n-1} \hspace{-0.5cm} \kappa^{-1}\Pt (s,s') + n (1-h(s))
            \end{flalign*}
            Observe that $\lim_{x \to +\infty} x \kappa^{x} =0$. So let
            $B = \sup_{x \ge 0} x \kappa^{x} \ge \kappa$. Using
            Corollary~\ref{coro:stronger},
            \begin{flalign*}
                   & \calL(s) - \sum_{s' \in S'}\Pt (s,s') \cdot \calL(s')                      \\
              \ge~ & -B \Pt^+(s) + \frac{1}{\kappa} \Pt^-(s)                                    
              + n
              \frac{2p-1}{(1-p)p}\cdot (\Pt^-(s)+\Pt^+(s)) \cdot (p^+-p)                        \\
              \ge~ & -B\Pt^+(s) + n  \frac{2p-1}{(1-p)p}\cdot (\Pt^-(s)+\Pt^+(s)) \cdot (p^+-p) \\
              \ge~ & \Pt^+(s)(-B + n  \frac{2p-1}{(1-p)p}\cdot (p^+-p))
            \end{flalign*}
            It is then sufficient to define $N_1$ such that
            $-B + N_1 \frac{2p-1}{(1-p)p}\cdot (p^+-p)\geq 0$.
    \end{description}
    The condition of
    Theorem~\ref{theorem:cnsdec} holds for all states. Thus,
    $\calL^{-1}([0,N_1]) = \alpha^{-1}([0,N_1]) \cup \{s_-\}$ is then a
    finite attractor of $\calC'$, which concludes the proof.
  \end{proof}
\fi

This theorem shows that the existence of a \emph{Lyapunov} state
function $\calL$ for some set $R$ ensures that $R$ is an attractor and
that the expected time to reach it with an explicit upper bound (given
in the proof) depending on the value of $\calL$ for the initial state.

\ifARXIV
  \theoremfoster*
  \begin{proof}
    W.l.o.g. we assume that all states in $R$ are absorbing.  Pick some
    $s\in S$. Define $X_n(s)$ as the random state at time $n$ when
    starting from $s$ (that is, $X_n^{\calC,s}$) and $T_{s,R}$ the
    random time (in $\bbN \cup \{\infty\}$) to reach $R$ from $s$.
    Observe that:
    $\Es\left(T_{s,R}\right)=\sum_{n\in \bbN} \Pr\left(X_n(s)\notin
      R\right)$.

    On the other hand, the inequality satisfied by $\calL$ can be
    rewritten as follows. For $Y$ random variable over $ S \setminus R$,
    $\Es(\calL(X_{1}(Y))-\calL(Y))\leq -\varepsilon$.

    Let $n\in \bbN$. Since states of $R$ are absorbing,
    \begin{align*}
           & \Es\left(\calL(X_{n+1}(s))-\calL(X_{0}(s))\right)                                                                  \\
      =    & \sum_{k\leq n} \Es\left(\calL(X_{k+1}(s))-\calL(X_{k}(s))\cdot{\bf 1}_{X_k(s)\notin R}\right)                      \\
      =    & \sum_{k\leq n} \Es\left(\calL(X_{k+1}(s))-\calL(X_{k}(s))\mid X_k(s)\notin R\right).\Pr\left(X_k(s)\notin R\right) \\
      \leq & -\varepsilon\sum_{k\leq n}\Pr\left(X_k(s)\notin R\right)
    \end{align*}
    Since $\calL$ is nonnegative and $\Es(\calL(X_{0}(s))) = \calL(s)$,
    one gets:
    $\sum_{k\leq n}\Pr(X_k(s)\notin R) \leq
      \frac{\calL(s)}{\varepsilon}$.  Letting $n$ go to $\infty$, one gets
    $\Es(T_{s,R})\leq \frac{\calL(s)}{\varepsilon}$, which concludes the
    proof.
  \end{proof}
\fi

The next proposition shows that choosing $\calW^{p}$ as an abstraction
with $\frac 1 2<p<p^+$ ensures decisiveness of $\calC'$ and finite expected
time for statistical model checking due to the previous theorem.

\ifARXIV
  \propdecisivefinite*
  \begin{proof}
    Let
    $ \hat{p}=\inf_{s \in \lambda^{-1}\left(]N_0,\infty[\right)}
      \Pt^+(s)$.  Due to Proposition~\ref{prop:layereddecisive}, we
    already know that $\calC'$ is decisive w.r.t. $T$. It remains to
    establish that the expected time to reach
    $T \cup \Av_{\calC'}(T) = T \cup \{s_-\}$ is finite. For every
    $s\in S\setminus \Av(F)$ with $\alpha(s)>N_0$,
    $0<\hat{p} \le \Pt^+(s)$, hence $\Pt^=(s)<1$. Therefore, using
    Corollary~\ref{coro:stronger}, for all
    $s \in S' \setminus \{s_-\} = S \setminus \Av_{\calC}(F)$ with
    $\alpha(s)>N_0$,
    \[
      \begin{array}{lcl}
        \calL(s) - \sum_{s' \in S'}\Pt (s,s') \cdot \calL(s')
         & \ge & \Pt^+(s) \cdot \left(-B + n  \cdot \frac{2p-1}{(1-p)p}\cdot (p^+-p)\right) \\
         & \ge & \hat{p} \cdot \left(-B
        + n  \cdot \frac{2p-1}{(1-p)p}\cdot (p^+-p)\right)
      \end{array}
    \]
    Let $N_1\geq N_0$ be such that
    $\hat{p} \cdot \left(-B + N_1 \cdot \frac{2p-1}{(1-p)p}\cdot (p^+-p)
      \right)\geq 1$ and
    $R= \calL^{-1}([0,N_1]) = \alpha^{-1}([0,N_1]) \cup \{s_-\}$. Then
    the condition of Theorem~\ref{foster} holds with
    $\varepsilon=1$. Applying it, the expected time to reach $R$
    from $s \in S \setminus \Av(F)$ with $\alpha(s)>N_1$ is finite and
    bounded by $\calL(s)=\alpha(s)$.

    It remains to establish that the expected time to reach $T$ from
    every state is finite. We fix an initial state $s_0 \in S'$ and we
    consider the infinite random sequence $(\gamma(n))_{n\in \nat}$,
    defined inductively as follows:
    $\gamma(0)=\min \left\{k \mid X_k(s_0)\in R\right\}$ and
    $\gamma(n+1)=\min \left\{k>\gamma(n) \mid X_k(s_0)\in R\right\}$;
    those are the successive times of visits in $R$.  Since $R$ is an
    attractor, this sequence is defined almost everywhere.  Let
    $h_{\max}=\max\left\{\calL(s')\mid \exists s \in R\ \text{s.t.}\
      \Pt'(s,s')>0\right\}$ the maximal level that can be reached in one
    step from $R$. Due to the previous paragraph, for all $n$ and
    $s\in R$,
    $\Es\left(\gamma(n+1)-\gamma(n)\mid X_{\gamma(n)}=s\right) \leq
      1+h_{\max}$ and $\Es\left(\gamma(0)\right) \leq \calL(s_0)$.

    Define 
    $\tilde{T} = T \cup \Av_{\calC'}(T) = T \cup \{s_-\}$,
    $\tau^{\calC',s_0,\tilde{T}}$ the (random) time to reach $\tilde{T}$
    from $s_0$ in $\calC'$, $\ell_{\max}$ as the maximal length of a shortest path
    from $s\in R$ to $T\cup \Av_{\calC'}(T)$ and $p_{\min}$ the
    minimal probability of these paths.  Then:
    \[
      \begin{array}{lcl}
        \multicolumn{3}{l}{\Es\left(\tau^{\calC',s_0,\tilde{T}}\right)} \\
        & =    & {\displaystyle
            \Es(\gamma(0))+ \sum_{n\in \nat}
            \Es\left(\gamma(n+1)-\gamma(n)\mid
            \bigwedge_{m\leq n} X_{\gamma(m)}\notin \tilde{T} \right) \cdot
            \Pr\left( \bigwedge_{m\leq n}
        X_{\gamma(m)}\notin \tilde{T} \right)}                                                                                                \\
                                                    & =    & {\displaystyle\Es(\gamma(0))+ \sum_{n\in \nat}
            \Es\left(\gamma(n+1)-\gamma(n)\mid
            X_{\gamma(n)}\notin \tilde{T} \right) \cdot
            \Pr\left(
        X_{\gamma(n)}\notin \tilde{T} \right)}                                                                                                \\
                                                    & \leq & {\displaystyle\calL(s_0)+ (1+h_{\max})\sum_{n\in \nat}
        \Pr\left(
        X_{\gamma(n)}\notin \tilde{T} \right)}                                                                                                \\
                                                    & =    & {\displaystyle\calL(s_0)+ (1+h_{\max})\sum_{n\in \nat}\sum_{0\leq j<\ell_{\max}}
        \Pr\left(
        X_{\gamma(n\ell_{\max}+j)}\notin \tilde{T} \right)}                                                                                   \\
                                                    & \leq & {\displaystyle\calL(s_0)+ (1+h_{\max})\ell_{\max}\sum_{n\in \nat}
        \Pr\left(
        X_{\gamma(n\ell_{\max})}\notin \tilde{T} \right)}                                                                                     \\
                                                    & \leq & {\displaystyle\calL(s_0)+ (1+h_{\max})\ell_{\max}\sum_{n\in \nat}
        (1-p_{\min})^n<\infty}
      \end{array}
    \]

  \end{proof}
\fi

\section{Details on the implementation presented in Section~\ref{subsec:impl}}

\subsection{Data-structure for exact summation}

Algorithm~\ref{algo:approx} heavily relies on the capacity to accurately sum
probabilities of very different magnitudes a large number of times.  Indeed in
early versions of the implementation, we have observed that without refined
dedicated summation algorithms, the program does not terminate.  Some
methods exist to improve the accuracy of summation like Kahan summation
algorithms~\cite{10.1145/363707.363723} but are not sufficient in our setting.
So we propose a data structure with better accuracy when summing up positive
values, at the cost of increased memory consumption and time.

We present our data structure in the context of values in the interval
$[0,1]$, which is sufficient for probabilities. It encodes such a
value $r$ as a table of $512$ integers (each encoded on 64 bits) such
that the content each cell $c[i]$ represents a floating point value
(float) with the content of the cell being the mantissa and the index
of the cell $i$ being the exponent. The value is encoded as the sum of
the floats encoded by the cells, i.e. $r=\sum_{i\leq 512} c[i]2^{-i}$.

\ifARXIV
  \SetKwFunction{exponent}{exponent}%
  \SetKwFunction{mantissa}{mantissa}%
  \SetKwFunction{buildFloat}{buildFloat}%

  We use three functions commonly available in any programming language
  for manipulating floats: $\exponent$, $\mantissa$ and $\buildFloat$
  which respectively extracts the exponent of a float as an integer,
  extracts the mantissa as an integer and builds a float given an
  exponent and a mantissa. Given a float $f$, they satisfy
  $f = \buildFloat(\exponent(f), \mantissa(f))$.

  The addition of a float $x$ to a table $T$ encoding a number is
  performed by Algorithm~\ref{algo:fldata}. The float is broken down
  into its exponent and mantissa i.e., $x= m 2^{-u}$ with
  $ m\leq 2^{52}$ and $1\leq u \leq 512$.  There are two cases : first
  (line 4) if $T[u]=0$ then $T[u]$ is set to $m$.  Otherwise the float
  stored at index $u$ is built ($y = T[u] 2^{-u}$) and $x$ and $y$ are
  added $z=x+y-\textit{err}$. The first subcase corresponds to the
  absence of overflow ($\textit{err}=0$ implying $\exponent(z)=u$)
  during this addition (line 9), then $T[u]=\mantissa(z)$. In case of an
  overflow occurs (line 11) then only the mantissa of the error is
  stored back in the table (i.e., $T[u]= \mantissa(\textit{err})$) and
  Algorithm~\ref{algo:fldata} is called recursively on $z$ (with
  $\exponent(z)=u+1$).  In the worst case there are recursive calls over
  the whole range of $T$.

  \begin{algorithm}
    \SetKwFunction{FRecurs}{add}%
    \SetKwInOut{Input}{input}\SetKwInOut{Output}{output}
    \SetKwProg{Fn}{def}{\string:}{}
    \Fn(){\FRecurs{T,x}}{
      \Input{$T$ a table representing a probability, $x$ a float
        representing a probability}
      \Output{None, the table $T$ is updated inplace.}
      $e := \exponent{x}$\;
      $m := \mantissa{x}$\;
      \uIf{$T[e] = 0$}{
        $T[e] \leftarrow m$\;
      }\Else{
        $y:= \buildFloat(e,T[e])$\;
        $z:= x+y$; \tcc{a numerical error may occur here}
        \uIf{ $\exponent{z}=e$ }{
          $T[e] := \mantissa{z}$\;
        }\Else{
          $T[e] := \mantissa{(z-y)-x}$; \tcc{the error is compensated here}
          $\FRecurs{T,z}$\;
        }    }
    }
    \caption{Data-structure encoding probabilities\label{algo:fldata}
      with accurate summation.}
  \end{algorithm}

\fi

\subsection{Heap with update}\label{subsec:heap}
Algorithm \ref{algo:approx} requires a data structure storing the set
of states that have been visited with the probability and likelihood
of the path reaching them. This data structure maps states to two real
numbers representing the probability and the likelihood of the state.
It requires to support three operations:
\begin{itemize}[nosep]
  \item insertion of a new mapping $s \mapsto (p,l)$ of a state to its
        probability and likelihood in the data structure;
  \item removing and retrieving the mapping with maximal probability;
  \item given a state $s$ updating the probability and likelihood of
        this state.
\end{itemize}
Such a data structure can be implemented with a heap and a hash table,
which points to the node in the heap allowing update.  All operations
are performed in logarithmic time w.r.t. to the number of elements in
the data structure.

\subsection{Implementation of the numerical algorithm for decisive
  Markov chains}
\label{subsec:num}

\ifARXIV
  Algorithm~\ref{algo:algodec} is a specialization of
  Algorithm~\ref{algo:approx} where the function to evaluate is the
  likelihood function of an importance sampling defined via an
  abstraction. As likelihood of an abstraction is a monoidal function of
  the path, the algorithm merges paths leading to the same state and
  only stores for each state the probability to reach it and its
  likelihood. A natural implementation of Algorithm~\ref{algo:approx}
  would have been using a queue. Using a heap where states are ordered
  by their probability to be reached and merging states as explained
  (and implemented in Algorithm~\ref{algo:algodec}) represents a large
  improvement.  The data structure for the heap is described in
  section~\ref{subsec:heap}. The merge operation is done on line 19, the
  probability are added while a weighted average is taken for the
  likelihood.

  \begin{algorithm}
    \SetKwInOut{Input}{input}\SetKwInOut{Output}{output}
    \SetKwData{Data}{data} \Input{$\calC=(S,\Pt)$ a Markov
      chain, $s_0 \in S$ a state,    $\calC'=(S \setminus \Av_{\calC}(F) \uplus \{s_-\},\Pt')$ a
      biased Markov chain of $(\calC,T,F)$, $\varepsilon>0$ a precision}
    \Output{An interval of width $\varepsilon$}
    \Data{$H$ an ordered mapping between $S$ and
      $([0;1]\times \mathbb{R}^+)$.  $p_{\text{fail}}$ and
      $p_{\text{succ}}$ two data structures encoding float with exact
      summation.}\;
    $H := \{ s_0 \rightarrow (1.0,1.0)\} $;
    $p_{\text{fail}}:=0$; $p_{\text{succ}}:=0$; $e:=0$\;
    \While{
    $H \neq \emptyset \wedge (1.0 - (p_{\text{fail}} +
      p_{\text{succ}}) > \varepsilon)$ } {
    $ s\rightarrow (w,L) := \textsf{pop}\_\textsf{max}(H)$\;
    \For { $s'$ s.t. $\Pt'(s,s') >0$ } {
      $L' := \frac{\Pt(s,s')}{\Pt'(s,s')} \cdot L$ ;
      $w' := w \cdot \Pt'(s,s')$\;
      \lIf{$s' \in T$}{ $\textsf{add}(p_{\text{succ}}, w' )$; $\textsf{add}(e, w'L' )$
      }\lElseIf{$s' \in T_- \cup \{s_-\}$}{
        $\textsf{add}(p_{\text{fail}}, w' )$ }
      \Else{ \uIf{
          $H[s']\neq \bot$}{ $(w'',L'') := H[s']$;
          $\textsf{update}(H, s' \rightarrow ( w'+w'', \frac{w' \cdot L'+ w''
              \cdot L''}{w'+w''}))$}
        \lElse{
          $\textsf{insert}(H, s' \rightarrow ( w', L'))$ } } } }
    $\Return([e-\varepsilon/2,e+\varepsilon/2])$
    \caption{Algorithm solving the \EvalEV problem\label{algo:algodec}}
  \end{algorithm}

  \begin{proposition}
    Algorithm~\ref{algo:algodec} terminates when applied on a decisive Markov chain.
  \end{proposition}
  \begin{proof}
    Assume that there exists some decisive Markov chain $\calC'$ (not
    necessarily obtained by importance sampling) on which
    Algorithm~\ref{algo:algodec} does not terminate.  We will establish
    that for all $d\in \nat$, the execution visits all vertices of
    $\text{Tr}_{\le d}$, the prefix of depth $d$ of $\text{Tr}$ the
    computation tree of $\calC'$. Since (by decisiveness of $\calC'$)
    there exists some $d$ such that the sum of the probability of the
    successful and lost paths of length at most $d$ is greater or equal
    than $1-\varepsilon$ implying termination, a contradiction.

    Towards a contradiction, assume that there exists some $d$ such that
    at least one vertex of $\text{Tr}_{\le d}$ is not visited. This
    implies that there exists some vertex $s$ of $\text{Tr}_{\le d}$
    that has been inserted in $H$ but not visited. Let $w$ be the weight
    (i.e. the probability of the path that has reached $s$) associated
    with $s$ when inserted in $H$. Observe that this weight can only be
    increased later.  Consider $d'\geq d$ such that the sum of the
    probabilities of the successful and lost paths of length at most
    $d'$ is larger than $1-w$. Since $\text{Tr}_{\le d'}$, there exists
    some round $r$ such that the execution does not visit anymore a
    vertex of $\text{Tr}_{\le d'}$.  Since the vertex $s$ belongs to $H$
    with weight larger than or equal to $w$, some vertex in
    $\text{Tr}_{\le d'}$ has to be selected in round $r+1$, which yields
    a contradiction.
  \end{proof}

\fi

\end{document}